\documentclass[journal]{IEEEtran}

\usepackage{graphicx}
\usepackage{subfigure}
\usepackage{amsmath,amsthm,amsfonts,amssymb}
\usepackage{cite}
\usepackage{bm}
\usepackage{bbm}
\usepackage{url}
\usepackage{array}
\usepackage{multirow}
\usepackage{booktabs}
\usepackage[table,xcdraw]{xcolor}
\usepackage{soul}
\usepackage{enumitem}
\usepackage{verbatim}
\usepackage{mathrsfs}
\usepackage{algorithm}
\usepackage[noend]{algpseudocode}
\usepackage[english]{babel}

\theoremstyle{plain}
\newtheorem{thm}{Theorem}
\newtheorem{lem}[thm]{Lemma}
\newtheorem{prop}[thm]{Proposition}
\newtheorem{cor}[thm]{Corollary}

\newtheorem{rem}{Remark}

\begin{document}

\addtolength{\textheight}{0.04in} % longer
\addtolength{\voffset}{-0.02in}

\addtolength{\textwidth}{0.06in} % wider
\addtolength{\hoffset}{-0.03in}

\title{Stepped Frequency Division Multiplexing: A Jump-Free Continuous-Time AFDM Waveform}

\author{Yewen Cao and Yulin Shao
\thanks{The authors are with the Department of Electrical and Computer Engineering, The University of Hong Kong, Hong Kong, China (e-mails: caoyewen3@gmail.com, ylshao@hku.hk).}
}

\maketitle

\begin{abstract}
Affine frequency division multiplexing (AFDM) has emerged as a promising modulation scheme for doubly selective channels, but its canonical continuous-time realization, referred to herein as piecewise continuous AFDM (PC-AFDM), has been observed to exhibit high out-of-band emission (OOBE) whose mechanism has not been analytically characterized.
This paper shows that the underlying cause is frequency wrapping, which introduces internal envelope jumps between AFDM sampling instants and generates a high-frequency spectral tail distinct from ordinary block truncation.
To eliminate these discontinuities without altering the inverse discrete affine Fourier transform (IDAFT) output sequence, we propose stepped frequency division multiplexing (SFDM). In SFDM, the instantaneous frequency is kept constant at the midpoint of the wrapped chirp within each sampling interval, while the phase is continuously accumulated across interval boundaries. We prove that, under continuous phase accumulation and without additional phase correction, the midpoint choice is the unique sample-preserving choice for arbitrary chirp-rate parameter. The resulting waveform is continuous within each AFDM block, reduces OOBE, and preserves the standard AFDM modulation matrix, guard-interval structure, and receiver processing. Moreover, under fractional-delay propagation, SFDM mitigates the receiver sensitivity that arises when delayed sampling points fall near wrapping-induced discontinuities in PC-AFDM.
Numerical results verify the theoretical tail coefficients, demonstrate OOBE reduction, and show improved receiver robustness in the high-percentile and worst-case regimes. These findings establish SFDM as a spectrally cleaner and more reliable physical layer for AFDM systems.
\end{abstract}


\section{Introduction}
\label{sec:introduction}

Doubly selective channels have motivated renewed interest in modulation schemes that spread data symbols over time, frequency, or delay-Doppler domains \cite{liu2004orthogonal,Matz2011TimeVaryingChannels,shao2021federated,Wu2011OversampledOFDM}. Representative examples include orthogonal time-frequency space (OTFS) modulation \cite{Hadani2017OTFS,Raviteja2018OTFS}, orthogonal chirp division multiplexing (OCDM) \cite{Ouyang2016OCDM}, Doppler resilient orthogonal signal division multiplexing (D-OSDM) \cite{Ebihara2016OSDM}, and affine frequency division multiplexing (AFDM) \cite{Bemani2021AFDM,Bemani2023AFDM}. Among them, AFDM uses the discrete affine Fourier transform (DAFT) to generate chirp basis functions whose parameters can be adjusted according to the delay-Doppler spread of doubly selective channels. Recent studies have further investigated AFDM for integrated sensing and communications (ISAC) \cite{10769778,cao2025agile}, channel estimation \cite{yang2026delay}, radar parameter estimation \cite{Ranasinghe2025AFDMISAC}, and ambiguity function analysis.

A shared thread running through these works is that they treat AFDM as a discrete modulation transform: an AFDM block is generated by the inverse DAFT (IDAFT), which produces a sequence of $N$ Nyquist-rate samples.
When the channel input-output relation is purely discrete, this view is adequate. However, once the modulation block leaves the digital baseband and becomes a physical continuous-time waveform, the IDAFT samples no longer tell the full story. Infinitely many continuous-time waveforms can interpolate the same $N$ samples, and these realizations differ in their inter-sample trajectories, their spectral containment, and, crucially, in how they respond to fractional-delay propagation. In other words, equality at the AFDM sampling instants is insufficient to characterize the transmitted waveform.

This leads to the following question: \emph{Given a discrete AFDM block defined via the IDAFT, how should one realize it as a continuous-time waveform while retaining the exact IDAFT output sequence, and what are the spectral and receiver consequences of different choices?}

The canonical continuous-time AFDM waveform, denoted herein as piecewise continuous AFDM (PC-AFDM), is formulated by the proposers of AFDM in \cite{Bemani2024ISAC}. In this construction, the instantaneous frequency of each chirp subcarrier is folded into the nominal bandwidth through frequency wrapping, and a phase correction is applied so that the waveform values at the sampling instants coincide with the IDAFT output sequence. Variants of this wrapped chirp realization have been adopted in recent AFDM studies on sensing waveform analysis, ambiguity function analysis, matched filtering, and continuous-time channel modeling \cite{Bemani2024ISAC,11173628,11489295}. However, the ambiguity function analysis in \cite{11173628} observed unexpectedly high out-of-band emission (OOBE) from continuous-time AFDM signals. The mechanism behind this behavior has not been analytically characterized. This motivates a systematic examination of the wrapped chirp realization itself, before external spectral-shaping measures are applied.

Our analysis shows that the observed OOBE is caused by internal discontinuities in PC-AFDM. The frequency wrapping operation partitions the time axis into segments with constant wrap count. At each wrapping instant, the phase correction that preserves the discrete AFDM block can introduce an abrupt jump in the complex envelope between sampling instants. These jumps disappear only for a discrete set of chirp-rate parameters. In the generic case, they create a high-frequency spectral tail beyond the nominal occupied bandwidth. This tail is not the ordinary endpoint truncation effect of a finite block. It is an additional contribution caused by wrapping-induced jumps inside the block.

This issue is different from conventional spectral shaping \cite{diez2019generalized}. Pulse shaping, windowing, and filtering are standard tools for reducing spectral leakage \cite{Brandes2006CancellationCarriers,vanDeBeek2009NContinuousOFDM,shao2024theory,Mahmoud2008AST,You2014OOBE}, but they operate by modifying the transmitted waveform or the observed samples. A continuous window applied to PC-AFDM can only scale an internal jump by the window value at the jump location. It cannot eliminate a nonzero jump unless the window vanishes there. This observation motivates a different construction that removes the internal jumps during waveform realization, while still preserving the discrete AFDM block.

To this end, this paper proposes stepped frequency division multiplexing (SFDM), a continuous-time AFDM realization that preserves the IDAFT output sequence and maintains a continuous complex envelope. The term ``stepped'' refers to the following construction. Within each sampling interval, the instantaneous frequency is held constant at the midpoint value of the wrapped affine chirp. The phase is accumulated continuously across adjacent intervals. This construction alters only the waveform between sampling instants. Its values on the sampling instants remain exactly those of the IDAFT output sequence. Within the class in which each sampling interval uses one representative frequency and the phase is accumulated continuously, we prove that the midpoint choice is the unique one that preserves the IDAFT output sequence uniformly over the chirp-rate parameter. SFDM is not the only possible continuous-time AFDM realization. It is, however, a principled candidate that retains the wrapped instantaneous frequency interpretation while eliminating the internal envelope jumps of PC-AFDM.
% 这句话想表达意思不明。
% 这句话先说："SFDM不唯一"，后面原理上应该跟着"但是他怎么怎么好"。但是，现在你后面跟的是陈述句，重复陈述了一下SFDM的性质
% 原本我写的应该是 "SFDM不唯一"+"但是他怎么怎么好"这个结构，强调的是它的好处。
% 你现在好像就是在陈述"SFDM不唯一"，and 重复陈述SFDM的性质？前者贬低SFDM，后者显得重复

The choice of continuous-time realization also affects the sampled receive signal \cite{Wu2011OversampledOFDM}. When every propagation delay is an integer multiple of the sampling period, PC-AFDM and SFDM lead to the same sampled channel matrix. The receiver observes only the sampling instants where the two waveforms coincide. The situation changes under fractional delays, which are common in wideband and sensing receivers. In this case, the receiver samples the transmitted waveform between grid points, and different continuous-time realizations generally lead to different sampled input-output relations \cite{Laakso1996FractionalDelay,11489295,221081}. SFDM therefore affects not only out-of-band leakage, but also the sampled receive signal under fractional delay propagation.

Building on this framework, the main contributions of this paper are summarized as follows.
\begin{enumerate}[leftmargin=0.45cm]
\item We reveal that the widely used continuous-time AFDM realization, PC-AFDM, exhibits internal envelope discontinuities. These discontinuities occur at frequency-wrapping boundaries for generic chirp rates. They have been overlooked in prior AFDM studies because they fall between Nyquist sampling instants and therefore leave the standard IDAFT output samples unchanged. We derive the exact temporal locations of these jumps, quantify their magnitudes in closed form, and establish the parametric condition under which they vanish.

\item We resolve the previously unexplained high OOBE observed in continuous-time AFDM by tracing it directly to the internal envelope discontinuities identified in this work. Through exact subcarrier spectra and a high-frequency asymptotic expansion, we prove that these jumps generate an additional spectral tail that decays inversely with the square of frequency and persists far beyond the nominal occupied bandwidth. The analysis thus uncovers the root cause of the anomalous OOBE and provides a quantitative tool for predicting and controlling out-of-band leakage in wrapped-chirp realizations.

\item We propose SFDM, a continuous-time AFDM realization that is structurally free of internal jumps. The key idea is to hold the instantaneous frequency constant at the midpoint value of the wrapped affine chirp within each sampling interval while accumulating phase continuously across interval boundaries. We prove that this midpoint rule is the unique sample-preserving choice uniformly over the chirp-rate parameter within the class of realizations with one constant instantaneous frequency over each sampling interval. SFDM thereby delivers three advantages over PC-AFDM.
\begin{itemize}
    \item First, its complex envelope is continuous, directly eliminating the internal-jump spectral tail without any sample distortion.
    \item Second, because the IDAFT samples are exactly preserved, SFDM requires no change to the discrete modulation matrix, guard-interval structure, or standard receiver processing.
    \item Third, under fractional-delay propagation, SFDM's smooth inter-sample trajectory suppresses the high-percentile and worst-case error vector magnitude (EVM) degradation that PC-AFDM suffers, yielding more robust receiver performance.
\end{itemize}

\item We clarify the relation between waveform realization and conventional postprocessing. Windowing reduces OOBE mainly by smoothing endpoint transitions at the cost of sample distortion, whereas SFDM eliminates internal wrapping-induced jumps with zero perturbation to the IDAFT samples. This establishes SFDM not as an alternative to windowing, but as a new, complementary degree of freedom in AFDM waveform design, compatible with all conventional post-processing methods.
\end{enumerate}
\section{Continuous-Time Realizations of AFDM}
\label{sec:ct_realizations}

To study the continuous-time behavior of AFDM, we start from its discrete model and analyze PC-AFDM as a reference realization. We show that while PC-AFDM reproduces the exact IDAFT output sequence, its frequency-wrapping structure generally introduces internal envelope jumps. This limitation leads us to propose SFDM, which preserves the IDAFT output sequence while removing these jumps and maintaining phase continuity.

\subsection{Discrete AFDM}
\label{subsec:discrete_afdm_model}

Consider an AFDM block with $N$ input symbols and total bandwidth $B$. Let $\bm{x}=[x[0],x[1],\ldots,x[N-1]]^T$ denote the data symbol vector in the DAF-domain. The block duration is $T=N/B$, and the sampling instants are $t_n=n/B$, $n=0,1,\ldots,N-1$. The corresponding discrete-time AFDM block is \cite{Bemani2021AFDM,Bemani2023AFDM}
\begin{equation}
s[n]=\frac{1}{\sqrt{N}}\sum_{m=0}^{N-1}x[m]
\exp\left\{j2\pi\left(c_2m^2+c_1n^2+\frac{mn}{N}\right)\right\},
\label{eq:discrete_afdm_samples}
\end{equation}
where $c_1$ and $c_2$ are the time-domain and DAF-domain AFDM chirp parameters, respectively. Throughout this paper, we focus on the case $c_1>0$. The case $c_1<0$ corresponds to the opposite chirp direction and can be treated by reversing the wrapping direction. The same mechanism applies, but the index sets in the spectral analysis need reversed inequalities. 

Eq. \eqref{eq:discrete_afdm_samples} can be written in a more compact form as $\bm{s}=\bm{A}_{\mathrm{IDAFT}}\bm{x}$, where
\[
{}[\bm{A}_{\mathrm{IDAFT}}]_{n,m}
=\frac{1}{\sqrt{N}}
\exp\left\{j2\pi\left(c_2m^2+c_1n^2+\frac{mn}{N}\right)\right\}.
\]
The discrete AFDM model in \eqref{eq:discrete_afdm_samples} specifies only the sequence $\{s[n]\}_{n=0}^{N-1}$ at the sampling instants. A continuous-time realization must reproduce these samples at $\{t_n\}_{n=0}^{N-1}$, but this condition does not uniquely determine the waveform over $0\le t<T$. Therefore, different continuous-time realizations can have the same discrete AFDM block but different trajectories between sampling instants and different spectra.

\begin{figure*}[!t]
    \centering
    \includegraphics[width=0.6\textwidth]{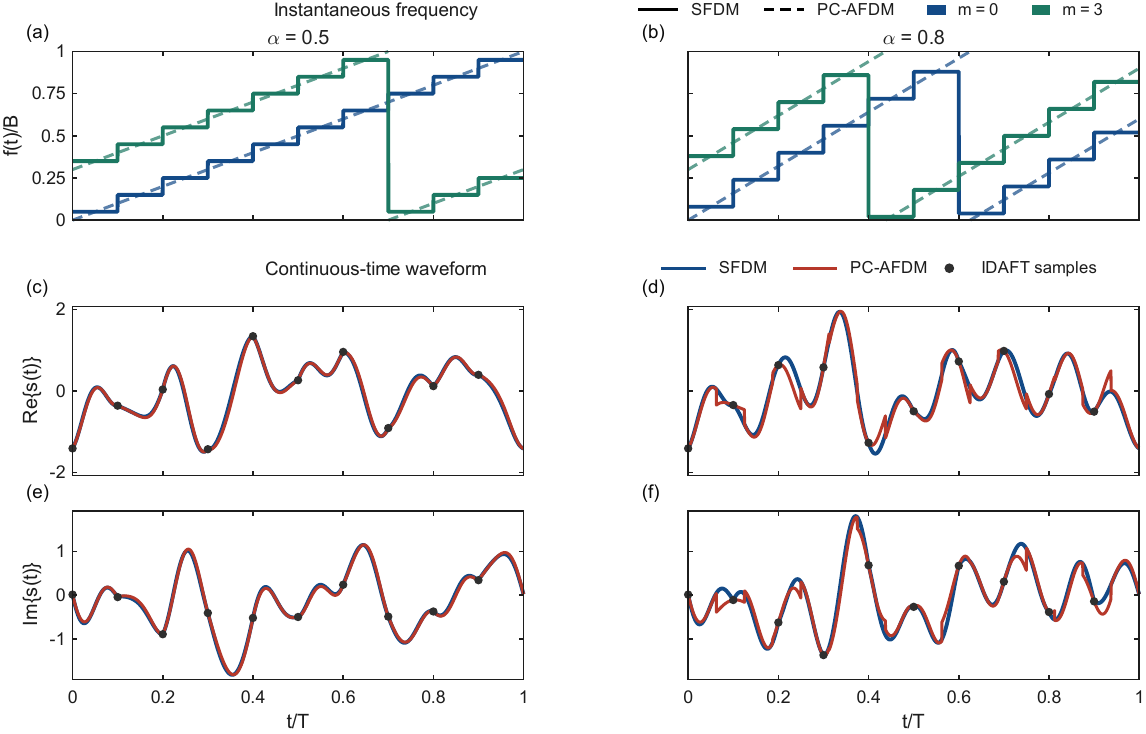}
    \caption{Continuous-time realizations of the same discrete AFDM block. The left and right columns correspond to $\alpha=0.5$ and $\alpha=0.8$, respectively. The top row shows the wrapped instantaneous frequency trajectories of representative subcarriers. The middle and bottom rows show the real and imaginary parts of the transmitted waveform. PC-AFDM and SFDM coincide at the sampling instants, while their trajectories between sampling instants differ. For $\alpha=0.8$, PC-AFDM contains internal envelope jumps, whereas SFDM remains continuous.}
    \label{fig:ct_realization_illustration}
\end{figure*}

\subsection{Piecewise Continuous AFDM}
\label{subsec:pc_afdm_realization}

The continuous-time AFDM waveform, as originally formulated by the proposers of AFDM in \cite{Bemani2024ISAC}, is obtained by rewriting their model in the notation used here. Its phase is specified piecewise according to the frequency wrapping index, and we therefore refer to it as the PC-AFDM waveform. This is the canonical continuous-time realization associated with the AFDM literature \cite{Bemani2024ISAC,11173628,11489295}. However, as we analyze in Section~\ref{sec:discontinuity_spectral_leakage}, this construction introduces internal discontinuities in the complex envelope at wrapping instants, which vanish only for special values of $\alpha$ and otherwise produce a pronounced high-frequency spectral tail. Recognizing this, we propose an alternative continuous-time realization, SFDM, that removes these internal jumps while preserving the IDAFT output sequence of the original AFDM formulation.

In the notation of this paper, the sampling period is $\Delta t=1/B$, the block duration is $T=N\Delta t=N/B$, and the continuous-time chirp rate is $K=2c_1B^2$. We define the normalized chirp-rate parameter as
\begin{equation}
\alpha\triangleq c_1N=\frac{KN}{2B^2}.
\label{eq:alpha_def}
\end{equation}
The $m$-th PC-AFDM phase can be written as
\begin{equation}
\phi_m^{(\mathrm{pc})}(t)=\frac{K}{2}t^2+\frac{m}{T}t-q_m^{(\mathrm{pc})}(t)Bt,
\label{eq:pc_phase_def}
\end{equation}
where
\begin{equation}
q_m^{(\mathrm{pc})}(t)=\left\lfloor\frac{Kt+m/T}{B}\right\rfloor
\label{eq:pc_wrap_count_def}
\end{equation}
is the frequency wrapping index. Whenever a wrapping instant is used in a jump expression, $t^{-}$ and $t^{+}$ denote the corresponding one-sided limits. The value assigned at an isolated wrapping instant does not affect the finite block spectrum.

For the $m$-th subcarrier, the unwrapped instantaneous frequency is $f_m^{(\mathrm{raw})}(t)=Kt+m/T$ for $0\le t<T$. On any open interval where $q_m^{(\mathrm{pc})}(t)$ is constant, the PC-AFDM phase is differentiable and its derivative is
\[
f_m^{(\mathrm{pc})}(t)
=\frac{d\phi_m^{(\mathrm{pc})}(t)}{dt}
=Kt+\frac{m}{T}-Bq_m^{(\mathrm{pc})}(t).
\]
Away from the wrapping instants, the wrapping term in \eqref{eq:pc_phase_def} confines the instantaneous frequency to $[0,B)$. The ordinary derivative is not used at the wrapping instants. The corresponding PC-AFDM waveform is
\begin{equation}
s^{(\mathrm{pc})}(t)
=\frac{1}{\sqrt{N}}\sum_{m=0}^{N-1}x[m]e^{j2\pi c_2m^2}
e^{j2\pi\phi_m^{(\mathrm{pc})}(t)}.
\label{eq:pc_tx_signal}
\end{equation}
At the sampling instants, the wrapping correction contributes an integer phase term because $Bt_n=n\in\mathbb{Z}$. Therefore,
\begin{equation}
e^{j2\pi\phi_m^{(\mathrm{pc})}(t_n)}
=e^{j2\pi(c_1n^2+mn/N)},\quad n=0,1,\ldots,N-1.
\label{eq:pc_sample_equivalence}
\end{equation}

As can be seen, PC-AFDM reproduces the standard IDAFT output sequence at the sampling instants. The waveform between adjacent sampling instants is determined by the wrapping operation in \eqref{eq:pc_wrap_count_def}.

\subsection{SFDM}
\label{subsec:sfdm_realization}

This paper puts forth a new continuous-time AFDM realization dubbed SFDM. It preserves the IDAFT output sequence while replacing the PC-AFDM trajectory between sampling instants with a waveform whose phase is continuous and whose instantaneous frequency is constant within each sampling interval. Let $I_n=[n/B,(n+1)/B)$, $n=0,1,\ldots,N-1$. On $I_n$, the instantaneous frequency is set to the midpoint value of the wrapped chirp:
\begin{equation}
f_m^{(\mathrm{step})}(t)
=K\frac{n+\tfrac{1}{2}}{B}+\frac{m}{T}
-Bq_m^{(\mathrm{step})}[n],
\label{eq:step_freq_def}
\end{equation}
where
\[
q_m^{(\mathrm{step})}[n]
=\left\lfloor\frac{K(n+\tfrac{1}{2})/B+m/T}{B}\right\rfloor.
\]
The floor rule uses the half-open frequency interval $[0,B)$. Thus, if the midpoint value lies on a wrapping boundary, the wrapped frequency is still uniquely specified.

Let $f_{m,n}^{(\mathrm{step})}$ denote the constant value of \eqref{eq:step_freq_def} on $I_n$. The phase is accumulated continuously according to
\begin{equation}
\begin{aligned}
\phi_m^{(\mathrm{step})}(0)&=0,\\
\phi_m^{(\mathrm{step})}\left(\frac{n+1}{B}\right)
&=\phi_m^{(\mathrm{step})}\left(\frac{n}{B}\right)
+\frac{1}{B}f_{m,n}^{(\mathrm{step})}.
\end{aligned}
\label{eq:step_phase_recursion}
\end{equation}
Hence, for $t\in I_n$,
\begin{equation}
\phi_m^{(\mathrm{step})}(t)
=\phi_m^{(\mathrm{step})}\left(\frac{n}{B}\right)
+f_{m,n}^{(\mathrm{step})}\left(t-\frac{n}{B}\right).
\label{eq:step_phase_piecewise}
\end{equation}
The SFDM waveform is
\begin{equation}
s^{(\mathrm{step})}(t)
=\frac{1}{\sqrt{N}}\sum_{m=0}^{N-1}x[m]e^{j2\pi c_2m^2}
e^{j2\pi\phi_m^{(\mathrm{step})}(t)}.
\label{eq:step_tx_signal}
\end{equation}
Thus, the SFDM complex envelope is continuous over $[0,T)$, although its instantaneous frequency may change at sampling interval boundaries.

\begin{thm}[Sampling equivalence of SFDM]
\label{thm:sfdm_sample_equivalence}
At the sampling instants, each SFDM subcarrier coincides with the corresponding standard IDAFT basis vector:
\begin{equation}
e^{j2\pi\phi_m^{(\mathrm{step})}(t_n)}
=e^{j2\pi(c_1n^2+mn/N)}.
\label{eq:step_sample_equivalence}
\end{equation}
Moreover, consider the class of realizations in which the instantaneous frequency on $I_n$ is evaluated at the representative time $(n+\theta)/B$:
\begin{equation}
\begin{aligned}
f_{m,\theta}(t)
&=K\frac{n+\theta}{B}+\frac{m}{T}
-Bq_{m,\theta}[n],\quad t\in I_n,\\
q_{m,\theta}[n]
&=
\left\lfloor
\frac{K(n+\theta)/B+m/T}{B}
\right\rfloor,\quad \theta\in[0,1).
\end{aligned}
\label{eq:theta_representative_def}
\end{equation}
The phase is accumulated continuously and no additional deterministic phase correction is used. Within this class, the midpoint choice $\theta=1/2$ is the unique choice that preserves the IDAFT output sequence uniformly with respect to $\alpha$.
\end{thm}

\begin{proof}
Evaluating \eqref{eq:step_phase_piecewise} at $t_n=n/B$ gives
\[
\begin{aligned}
\phi_m^{(\mathrm{step})}(t_n)
&=\sum_{i=0}^{n-1}\frac{1}{B}
\left[K\frac{i+\tfrac{1}{2}}{B}+\frac{m}{T}
-Bq_m^{(\mathrm{step})}[i]\right]\\
&=\frac{K}{B^2}\sum_{i=0}^{n-1}\left(i+\frac{1}{2}\right)
+\frac{mn}{BT}-\sum_{i=0}^{n-1}q_m^{(\mathrm{step})}[i].
\end{aligned}
\]
Using $\sum_{i=0}^{n-1}(i+1/2)=n^2/2$, $T=N/B$, and $c_1=K/(2B^2)$, we obtain
\[
\phi_m^{(\mathrm{step})}(t_n)=c_1n^2+\frac{mn}{N}-\kappa_{m,n},
\]
where $\kappa_{m,n}=\sum_{i=0}^{n-1}q_m^{(\mathrm{step})}[i]\in\mathbb{Z}$. Exponentiating by $j2\pi(\cdot)$ establishes \eqref{eq:step_sample_equivalence}.

For \eqref{eq:theta_representative_def}, the same summation gives
\[
\begin{aligned}
\phi_{m,\theta}(t_n)
&=\frac{K}{B^2}\sum_{i=0}^{n-1}(i+\theta)
+\frac{mn}{BT}-\sum_{i=0}^{n-1}q_{m,\theta}[i]\\
&=c_1n^2+\frac{mn}{N}+c_1(2\theta-1)n
-\sum_{i=0}^{n-1}q_{m,\theta}[i].
\end{aligned}
\]
The final summation is integer-valued. Therefore, exact agreement with the IDAFT exponential for all $n$ and uniformly over the normalized chirp-rate parameter $\alpha$ requires $\alpha(2\theta-1)n/N$ to be an integer for every $n$. Since $\alpha$ can vary over a continuous range, this condition can hold only when $2\theta-1=0$, i.e., $\theta=1/2$.
\end{proof}

% \begin{rem}
% The term SFDM deliberately omits the word ``affine''. In PC-AFDM, the unwrapped instantaneous frequency of each subcarrier follows an affine function of time, and the wrapping operation preserves this affine relationship on each open interval of constant wrap count. The PC-AFDM waveform is genuinely an affine frequency division multiplexing waveform, exhibiting chirp segments of constant slope. SFDM, by contrast, replaces this affine frequency trajectory with a staircase of constant frequencies: on each sampling interval, the instantaneous frequency is held fixed, and the phase accumulates linearly. Between intervals, the frequency may jump to a new constant value. The piecewise constant nature is what ``stepped'' designates. Consequently, the time-frequency trajectory of SFDM is not piecewise affine but piecewise constant in frequency; the affine structure is deliberately sacrificed to achieve a globally continuous complex envelope. Calling it ``stepped affine'' would conflate two mutually exclusive descriptions of the instantaneous frequency trajectory. We therefore use SFDM to reflect this distinction.
% \end{rem}

\begin{rem}[On the naming of SFDM]
The term SFDM deliberately omits ``affine'' to avoid a conceptual contradiction. In PC-AFDM, the unwrapped instantaneous frequency of each subcarrier is affine in time, and this affine dependence is preserved on every open interval with constant wrap count, producing genuine chirp segments of constant slope. SFDM, by contrast, replaces this affine trajectory with a staircase of constant instantaneous frequencies: on each sampling interval the frequency is held fixed, the phase accumulates linearly, and the frequency changes only at interval boundaries. The affine time-frequency structure is intentionally set aside in order to obtain a globally continuous complex envelope. Describing this piecewise constant frequency waveform as ``stepped affine'' would conflate two mutually incompatible characterizations of the instantaneous frequency trajectory. The name SFDM reflects this deliberate distinction.
\end{rem}

Let $\bm{A}_{\mathrm{pc}}$ and $\bm{A}_{\mathrm{step}}$ denote the matrices obtained by sampling the PC-AFDM and SFDM basis functions at $\{t_n\}_{n=0}^{N-1}$, respectively. For $\xi\in\{\mathrm{pc},\mathrm{step}\}$, their entries are
\[
{}[\bm{A}_{\xi}]_{n,m}
=\frac{1}{\sqrt{N}}\exp\left\{j2\pi\left(c_2m^2+\phi_m^{(\xi)}(t_n)\right)\right\}.
\]
By \eqref{eq:pc_sample_equivalence} and Theorem~\ref{thm:sfdm_sample_equivalence},
\begin{equation}
\bm{A}_{\mathrm{pc}}\equiv\bm{A}_{\mathrm{step}}\equiv\bm{A}_{\mathrm{IDAFT}}.
\label{eq:matrix_equivalence}
\end{equation}
Thus, the sampled matrices of PC-AFDM and SFDM are both equal to the standard IDAFT modulation matrix. Their difference lies in how the waveform evolves between adjacent sampling instants. This difference affects envelope continuity, spectral behavior, and the sampled input-output relation after continuous-time propagation.

Fig.~\ref{fig:ct_realization_illustration} illustrates this distinction. The two realizations have identical values at the sampling instants, while their trajectories between adjacent samples are different. For $\alpha=0.5$, the phase reset in PC-AFDM corresponds to integer multiples of $2\pi$, and the wrapping-induced jumps vanish. For $\alpha=0.8$, PC-AFDM exhibits internal complex envelope jumps at wrapping instants, whereas SFDM remains continuous because its phase is accumulated across sampling intervals.

\begin{rem}
 The uniqueness statement in Theorem~\ref{thm:sfdm_sample_equivalence} applies only to the class in \eqref{eq:theta_representative_def}. It does not exclude other continuous-time realizations that reproduce the same discrete AFDM block, such as realizations based on interpolation, pulse shaping, or additional phase correction. This class is useful because it keeps the wrapped instantaneous frequency interpretation of AFDM and changes only the representative frequency used inside each sampling interval. It also avoids an external interpolation filter. Therefore, the result should be interpreted as a uniqueness result for this phase accumulation construction, not as a uniqueness result among all possible continuous-time AFDM waveforms.   
\end{rem}

\subsection{Local Relation Between PC-AFDM and SFDM}
\label{subsec:local_inter_sample_relation}

The spectral analysis in Section~\ref{sec:discontinuity_spectral_leakage} will show that PC-AFDM contains an additional OOBE contribution relative to SFDM.
Before focusing on wrapping-induced jumps, it is useful to quantify the local difference between PC-AFDM and SFDM on intervals where no wrapping occurs. On such intervals, PC-AFDM has a quadratic phase, whereas SFDM has a linear phase determined by the midpoint frequency. The following proposition shows that this local difference is uniformly bounded after a constant phase rotation.
% In any interval where the PC-AFDM wrap count remains constant, the phase trajectories of PC-AFDM and SFDM are both smooth within the interval, but their functional forms differ.
% PC-AFDM follows a quadratic phase (the original unwrapped chirp minus a constant offset), while SFDM follows a linear phase determined by the midpoint frequency.
% If this local approximation error were already large, it could dominate the OOBE difference, rendering the subsequent jump-focused analysis irrelevant.
% The purpose of this subsection is therefore to quantify the deviation between the two waveforms away from wrapping instants and to demonstrate that it is uniformly small.
% Hence, any significant spectral difference must originate from the wrapping-induced discontinuities, not from the smooth interpolation error.

To formalize this idea, define the complex envelope basis function of realization $\xi\in\{\mathrm{pc},\mathrm{step}\}$ as $g_m^{(\xi)}(t)=e^{j2\pi\phi_m^{(\xi)}(t)}$, $0\leq t< T$.
Because the two realizations share the same samples at $t_n=n/B$, a natural way to compare them locally is to examine the difference between $g_m^{(\mathrm{pc})}(t)$ and a phase-rotated version of $g_m^{(\mathrm{step})}(t)$ on each sampling interval where the PC-AFDM wrap count is constant. Proposition \ref{prop:local_pc_sfdm_relation} below provides a precise bound.

\begin{prop}[Local relation between PC-AFDM and SFDM]
\label{prop:local_pc_sfdm_relation}
Consider an interval $I_n$ on which the PC-AFDM wrap count of the $m$-th subcarrier is constant and equal to $q_{m,n}$. Then, for all $t\in I_n$,
\begin{equation}
\phi_m^{(\mathrm{pc})}(t)-\phi_m^{(\mathrm{step})}(t)
=C_{m,n}+\frac{K}{2}(t-t_{c,n})^2,
\label{eq:local_phase_relation}
\end{equation}
where $C_{m,n}$ is independent of $t$ and $t_{c,n}=(n+1/2)/B$. Consequently,
\begin{equation}
\left|g_m^{(\mathrm{pc})}(t)-e^{j2\pi C_{m,n}}g_m^{(\mathrm{step})}(t)\right|
\le \pi K(t-t_{c,n})^2
\le \frac{\pi\alpha}{2N}.
\label{eq:local_waveform_bound}
\end{equation}
\end{prop}

\begin{proof}
On $I_n$, the derivative of the PC-AFDM phase is $d\phi_m^{(\mathrm{pc})}(t)/dt=Kt+m/T-q_{m,n}B$. Since the wrap count is constant on $I_n$, the SFDM frequency is $f_{m,n}^{(\mathrm{step})}=K(n+1/2)/B+m/T-q_{m,n}B$. Therefore,
\[
\frac{d}{dt}\left[\phi_m^{(\mathrm{pc})}(t)-\phi_m^{(\mathrm{step})}(t)\right]
=K(t-t_{c,n}).
\]
Integrating over $I_n$ gives \eqref{eq:local_phase_relation}. The bound \eqref{eq:local_waveform_bound} follows from $|e^{jx}-1|\le |x|$ and $|t-t_{c,n}|\le 1/(2B)$.
\end{proof}

Proposition~\ref{prop:local_pc_sfdm_relation} shows that, on intervals without wrapping, the two basis functions differ only by a bounded local phase curvature term after a constant phase rotation. When $\alpha/N$ is small, this local difference is much smaller than the possible unit-order envelope jumps at wrapping instants. This supports the subsequent focus on wrapping-induced jumps as the dominant source of the additional high-frequency tail.

This local analysis justifies treating the smooth interpolation error as a secondary effect in the subsequent OOBE analysis: the spectral leakage contributed by these intervals cannot explain the strong high-frequency tail observed for PC-AFDM at generic values of $\alpha$.
With this local smoothness verified, we can safely attribute the OOBE advantage of SFDM to the removal of internal jumps, a topic to which we now turn.
\section{Discontinuity-Induced Spectral Tails}
\label{sec:discontinuity_spectral_leakage}

This section analyzes how the waveform between sampling instants affects the spectrum. The main result is that PC-AFDM can introduce internal jumps in the complex envelope. These jumps create an additional contribution to the high-frequency spectral tail, which is absent in SFDM. Both realizations still contain the finite block endpoint contribution. The tail analysis therefore isolates the effect of the internal jumps, while the full OOBE also depends on the spectrum near the nominal band.

\subsection{Spectral Measures for Continuous-Time AFDM}
\label{subsec:spectral_measures}

For $\xi\in\{\mathrm{pc},\mathrm{step}\}$, define the finite block spectrum as $S^{(\xi)}(f)=\int_0^T s^{(\xi)}(t)e^{-j2\pi ft}dt$. Using \eqref{eq:pc_tx_signal} and \eqref{eq:step_tx_signal}, and with $g_m^{(\xi)}(t)=e^{j2\pi\phi_m^{(\xi)}(t)}$, both waveforms can be written as
\begin{equation}
s^{(\xi)}(t)=\frac{1}{\sqrt{N}}\sum_{m=0}^{N-1}x[m]e^{j2\pi c_2m^2}g_m^{(\xi)}(t).
\label{eq:common_waveform_def}
\end{equation}
Assuming mutually independent, zero-mean, unit-variance symbols, the average energy spectral density (ESD) \cite{oppenheim1997signals} is
\begin{equation}
\Phi^{(\xi)}(f)=\mathbb{E}\left[|S^{(\xi)}(f)|^2\right]
=\frac{1}{N}\sum_{m=0}^{N-1}|G_m^{(\xi)}(f)|^2,
\label{eq:average_esd_def}
\end{equation}
where $G_m^{(\xi)}(f)=\int_0^T g_m^{(\xi)}(t)e^{-j2\pi ft}dt$. The averaging in \eqref{eq:average_esd_def} is only over the data symbols. The functions $g_m^{(\xi)}(t)$ are deterministic for fixed $N$, $B$, $c_1$, and $c_2$. The factor $e^{j2\pi c_2m^2}$ is a constant phase for the $m$-th subcarrier, so it does not affect $|G_m^{(\xi)}(f)|^2$ in the average ESD.

The normalized OOBE ratio with respect to the nominal occupied frequency interval $[0,B)$ is defined as
\begin{equation}
\eta_{\mathrm{OOBE}}^{(\xi)}
=\frac{1}{T}\int_{\mathbb{R}\setminus[0,B)}\Phi^{(\xi)}(f)df.
\label{eq:oobe_ratio_def}
\end{equation}
By Parseval's identity, $\int_{-\infty}^{\infty}\Phi^{(\xi)}(f)df=T$. The OOBE ratio in \eqref{eq:oobe_ratio_def} integrates the entire out-of-band region. It is therefore not determined only by the high-frequency tail coefficient derived later.

The following analysis separates two spectral contributions. A finite block always produces endpoint terms at $t=0$ and $t=T$, which are common to PC-AFDM and SFDM. The distinguishing term is the possible internal jump of the PC-AFDM complex envelope. We therefore first locate the internal wrapping instants and their jump magnitudes, and then relate these jumps to the high-frequency spectral tail.

\subsection{Wrapping-Induced Discontinuities in PC-AFDM}
\label{subsec:pc_wrapping_discontinuities}

We first identify where PC-AFDM may lose continuity. These points occur when the unwrapped affine instantaneous frequency crosses an integer multiple of the bandwidth. The normalized chirp-rate parameter $\alpha$ is defined in \eqref{eq:alpha_def}. To distinguish the time-varying wrap count $q_m^{(\mathrm{pc})}(t)$ from the integer boundary crossed by the unwrapped instantaneous frequency, let $r$ denote the wrapping boundary index. For $c_1>0$, the admissible internal boundary index set is
\begin{equation}
\mathcal{R}_m=\left\{r\in\mathbb{Z}:\frac{m}{N}<r<2\alpha+\frac{m}{N}\right\}.
\label{eq:wrap_index_set}
\end{equation}
The number of internal wrapping instants is
\[
J_m=|\mathcal{R}_m|
=\max\left\{0,\left\lceil 2\alpha+\frac{m}{N}\right\rceil-1\right\}.
\]

\begin{lem}[Internal wrapping instants and jump magnitudes]
\label{lem:internal_wrapping_geometry}
For the $m$-th PC-AFDM subcarrier, each $r\in\mathcal{R}_m$ gives an internal wrapping instant
\begin{equation}
t_{m,r}=\frac{rB-m/T}{K}=\frac{Nr-m}{2\alpha B}.
\label{eq:wrap_time}
\end{equation}
At an internal wrapping instant $t_{m,r}$, the jump satisfies
\begin{equation}
\begin{aligned}
\Delta g_{m,r}&=g_m^{(\mathrm{pc})}(t_{m,r}^{+})
-g_m^{(\mathrm{pc})}(t_{m,r}^{-})\\
&=g_m^{(\mathrm{pc})}(t_{m,r}^{-})
\left(e^{-j2\pi Bt_{m,r}}-1\right),
\end{aligned}
\label{eq:jump_def}
\end{equation}
and
\begin{equation}
|\Delta g_{m,r}|^2
=4\sin^2(\pi Bt_{m,r})
=4\sin^2\left(\frac{\pi(Nr-m)}{2\alpha}\right).
\label{eq:jump_magnitude}
\end{equation}
\end{lem}

\begin{proof}
The internal wrapping instants occur when the floor operator argument in \eqref{eq:pc_wrap_count_def} crosses an integer boundary, namely $Kt+m/T=rB$. Solving for $t$ gives \eqref{eq:wrap_time}. The condition $0<t_{m,r}<T=N/B$ is equivalent to $0<(Nr-m)/(2\alpha B)<N/B$, which yields \eqref{eq:wrap_index_set}. Counting the admissible integers gives the expression for $J_m$.

At a wrapping instant, the wrap count increases by one. Hence, the phase in \eqref{eq:pc_phase_def} changes by $-Bt_{m,r}$, giving
\[
g_m^{(\mathrm{pc})}(t_{m,r}^{+})
=g_m^{(\mathrm{pc})}(t_{m,r}^{-})e^{-j2\pi Bt_{m,r}}.
\]
This proves \eqref{eq:jump_def}. Since $|g_m^{(\mathrm{pc})}(t_{m,r}^{-})|=1$, we have $|\Delta g_{m,r}|^2=|e^{-j2\pi Bt_{m,r}}-1|^2=4\sin^2(\pi Bt_{m,r})$. Substituting \eqref{eq:wrap_time} gives \eqref{eq:jump_magnitude}.
\end{proof}

Lemma~\ref{lem:internal_wrapping_geometry} separates the occurrence of internal wrapping from the jump magnitude. The former is determined by the unwrapped affine instantaneous frequency. The latter depends on whether the phase reset at a wrapping instant corresponds to an integer multiple of $2\pi$.

\begin{cor}[Jump of the transmitted PC-AFDM waveform]
\label{cor:tx_waveform_jump}
Let $t_0$ be an internal wrapping instant of at least one PC-AFDM subcarrier, and let $\Delta g_m(t_0)=g_m^{(\mathrm{pc})}(t_0^+)-g_m^{(\mathrm{pc})}(t_0^-)$ if the $m$-th subcarrier wraps at $t_0$ and $\Delta g_m(t_0)=0$ otherwise. Then the jump of the transmitted PC-AFDM waveform at $t_0$ is
\begin{equation}
\Delta s^{(\mathrm{pc})}(t_0)
=
\frac{1}{\sqrt{N}}\sum_{m=0}^{N-1}
x[m]e^{j2\pi c_2m^2}\Delta g_m(t_0).
\label{eq:tx_waveform_jump}
\end{equation}
If the data symbols are mutually independent, zero mean, and unit variance, then
\begin{equation}
\mathbb{E}\left[|\Delta s^{(\mathrm{pc})}(t_0)|^2\right]
=
\frac{1}{N}\sum_{m=0}^{N-1}|\Delta g_m(t_0)|^2.
\label{eq:tx_waveform_jump_energy}
\end{equation}
\end{cor}

\begin{proof}
Equation \eqref{eq:tx_waveform_jump} follows directly from \eqref{eq:common_waveform_def}. For \eqref{eq:tx_waveform_jump_energy}, expand $|\Delta s^{(\mathrm{pc})}(t_0)|^2$. The cross terms vanish because the data symbols are mutually independent and zero mean, while the diagonal terms have unit variance.
\end{proof}

Corollary~\ref{cor:tx_waveform_jump} clarifies the relation between subcarrier jumps and the transmitted waveform. For a particular data block, jumps from different subcarriers may cancel at a specific time. In the average ESD considered in \eqref{eq:average_esd_def}, the jump contribution is determined by the squared jump magnitudes of the basis functions.

The next result identifies the parameter regimes in which the PC-AFDM basis functions remain continuous. It follows directly from the wrapping condition and the jump magnitude in Lemma~\ref{lem:internal_wrapping_geometry}.

\begin{thm}[Continuity condition for PC-AFDM basis functions]
\label{thm:pc_continuity_criterion}
The PC-AFDM basis functions are continuous over $[0,T)$ for all subcarriers if and only if one of the following two conditions holds:
\begin{equation}
\alpha\le \frac{1}{2N},
\label{eq:no_wrap_condition_alpha}
\end{equation}
or
\begin{equation}
\alpha=\frac{1}{2k},\quad k\in\mathbb{Z}_{>0}.
\label{eq:special_alpha_condition}
\end{equation}
\end{thm}

\begin{proof}
No internal wrapping occurs for any subcarrier if and only if $\mathcal{R}_m=\varnothing$ for all $m$. The most stringent condition is obtained at $m=N-1$, which gives $2\alpha+(N-1)/N\le 1$. This is equivalent to \eqref{eq:no_wrap_condition_alpha}. In this regime, each wrap count is constant over the whole block, and all PC-AFDM basis functions are continuous.

For $\alpha>1/(2N)$, internal wrapping occurs. In particular, $r=1$ belongs to $\mathcal{R}_{N-1}$, and the corresponding wrapping instant is $t_{N-1,1}=1/(2\alpha B)$. If all basis functions are continuous, the jump at this instant must vanish. By \eqref{eq:jump_def}, this requires $e^{-j2\pi Bt_{N-1,1}}=1$, or $1/(2\alpha)\in\mathbb{Z}$. Thus $\alpha=1/(2k)$ for some $k\in\mathbb{Z}_{>0}$.

Conversely, if $\alpha=1/(2k)$, then $Bt_{m,r}=(Nr-m)/(2\alpha)=k(Nr-m)\in\mathbb{Z}$ for any admissible internal wrapping instant. Therefore, every internal jump in \eqref{eq:jump_def} vanishes. This proves the necessity and sufficiency of \eqref{eq:no_wrap_condition_alpha} and \eqref{eq:special_alpha_condition}.
\end{proof}

Theorem~\ref{thm:pc_continuity_criterion} shows that the PC-AFDM basis functions are continuous only in special parameter regimes. Apart from the regime with no wrapping and the discrete values $\alpha=1/(2k)$, at least one subcarrier contains a nonzero internal jump in the complex envelope.

\subsection{High-Frequency Spectral Tail Analysis}
\label{subsec:high_frequency_tail_analysis}

We now connect the time-domain jumps to the spectral tail. For a finite duration piecewise smooth waveform, integration by parts shows that endpoints and jumps contribute the leading $1/f$ terms to the spectrum. The smooth parts of the waveform only affect higher-order terms. For fixed $N$, $B$, and $\alpha$, each subcarrier waveform is smooth on a finite number of intervals. A finite block always contributes endpoint terms at $t=0$ and $t=T$, whereas PC-AFDM may additionally contribute internal jump terms.

\begin{lem}[High-frequency expansion]
\label{lem:high_frequency_expansion}
As $|f|\to\infty$, the subcarrier spectra satisfy
\begin{equation}
\begin{aligned}
G_m^{(\mathrm{pc})}(f)
&=\frac{A_m^{(\mathrm{pc})}(f)}{j2\pi f}+O(f^{-2}),\\
G_m^{(\mathrm{step})}(f)
&=\frac{A_m^{(\mathrm{step})}(f)}{j2\pi f}+O(f^{-2}).
\end{aligned}
\label{eq:pc_high_freq_expansion}
\end{equation}
where
\begin{equation}
\begin{aligned}
A_m^{(\mathrm{pc})}(f)
&=g_m^{(\mathrm{pc})}(0^{+})
-g_m^{(\mathrm{pc})}(T^{-})e^{-j2\pi fT}\\
&\quad+\sum_{r\in\mathcal{R}_m}\Delta g_{m,r}e^{-j2\pi ft_{m,r}},
\end{aligned}
\label{eq:a_pc_def}
\end{equation}
and
\begin{equation}
A_m^{(\mathrm{step})}(f)
=g_m^{(\mathrm{step})}(0^{+})
-g_m^{(\mathrm{step})}(T^{-})e^{-j2\pi fT}.
\label{eq:a_step_def}
\end{equation}
\end{lem}

\begin{proof}
For PC-AFDM, partition $[0,T]$ by the ordered boundary sequence $\{\vartheta_{m,j}\}_{j=0}^{J_m+1}$. On each open segment, $g_m^{(\mathrm{pc})}(t)$ is smooth. Integration by parts gives
\[
\begin{aligned}
&\int_{\vartheta_{m,j}}^{\vartheta_{m,j+1}}
g_m^{(\mathrm{pc})}(t)e^{-j2\pi ft}dt\\
&=\frac{g_m^{(\mathrm{pc})}(\vartheta_{m,j}^{+})e^{-j2\pi f\vartheta_{m,j}}
-g_m^{(\mathrm{pc})}(\vartheta_{m,j+1}^{-})e^{-j2\pi f\vartheta_{m,j+1}}}{j2\pi f}\\
&\quad+\frac{1}{j2\pi f}\int_{\vartheta_{m,j}}^{\vartheta_{m,j+1}}
\frac{d g_m^{(\mathrm{pc})}(t)}{dt}e^{-j2\pi ft}dt.
\end{aligned}
\]
Summing over all segments telescopes the interior boundary terms except at the jump locations. This gives \eqref{eq:pc_high_freq_expansion} and \eqref{eq:a_pc_def}. The remaining integral is $O(f^{-2})$ after applying integration by parts once more on each smooth segment. For SFDM, the subcarrier is continuous over $[0,T)$ and smooth on each sampling interval, so the same argument produces only the endpoint terms in \eqref{eq:a_step_def}.
\end{proof}

The high-frequency expansion separates the common endpoint contribution from the additional internal jump contribution. Squaring the expansion and integrating over a high-frequency interval gives the following tail coefficients.

\begin{thm}[High-frequency spectral tail]
\label{thm:far_out_tail}
For each subcarrier $m$, as $F\to\infty$ with fixed waveform parameters,
\begin{equation}
\begin{aligned}
&\int_F^{\infty}|G_m^{(\mathrm{step})}(f)|^2df
= \frac{1}{2\pi^2F}+O(F^{-2}),\\
&\int_F^{\infty}|G_m^{(\mathrm{pc})}(f)|^2df \\
&\quad = \frac{1}{2\pi^2F}
+ \frac{1}{4\pi^2F}\sum_{r\in\mathcal{R}_m}|\Delta g_{m,r}|^2
+ O(F^{-2})\\
&\quad = \frac{1}{2\pi^2F}
+ \frac{1}{\pi^2F}\sum_{r\in\mathcal{R}_m}
\sin^2\left(\frac{\pi(Nr-m)}{2\alpha}\right)
+ O(F^{-2}).
\end{aligned}
\label{eq:subcarrier_tail}
\end{equation}
Averaging over the subcarriers gives the one-sided positive-frequency tails
\begin{equation}
\begin{aligned}
&\int_F^{\infty}\Phi^{(\mathrm{step})}(f)df
= \frac{1}{2\pi^2F} + O(F^{-2}), \\
&\int_F^{\infty}\Phi^{(\mathrm{pc})}(f)df
= \frac{1}{2\pi^2F} \\
&\qquad + \frac{1}{\pi^2NF}
\sum_{m=0}^{N-1}\sum_{r\in\mathcal{R}_m}
\sin^2\left(\frac{\pi(Nr-m)}{2\alpha}\right)
+ O(F^{-2}).
\end{aligned}
\label{eq:esd_tail}
\end{equation}
Furthermore, for the two-sided high-frequency tail $\mathcal{T}^{(\xi)}(F)=\int_{|f|>F}\Phi^{(\xi)}(f)df$, $F\ge B$, we have
\begin{equation}
\begin{aligned}
\mathcal{T}^{(\mathrm{step})}(F)
&=\frac{1}{\pi^2F}+O(F^{-2}),\\
\mathcal{T}^{(\mathrm{pc})}(F)
&=\frac{1}{\pi^2F}
+\frac{2}{\pi^2NF}\sum_{m=0}^{N-1}\sum_{r\in\mathcal{R}_m}
\sin^2\left(\frac{\pi(Nr-m)}{2\alpha}\right)\\
&\quad+O(F^{-2}).
\end{aligned}
\label{eq:two_sided_esd_tail}
\end{equation}
\end{thm}

\begin{proof}
By Lemma~\ref{lem:high_frequency_expansion},
\[
|G_m^{(\xi)}(f)|^2
=\frac{|A_m^{(\xi)}(f)|^2}{4\pi^2f^2}+O(f^{-3}),
\quad \xi\in\{\mathrm{pc},\mathrm{step}\}.
\]
For any fixed $\tau\neq0$, a single integration by parts gives
\[
\int_F^{\infty}\frac{e^{-j2\pi\tau f}}{f^2}df=O(F^{-2}),
\quad F\to\infty.
\]
Therefore, after integrating over $[F,\infty)$, only the squared magnitudes of the diagonal terms in $|A_m^{(\xi)}(f)|^2$ contribute to the $1/F$ coefficient. For SFDM, the two endpoint coefficients have unit magnitude, giving the first line of \eqref{eq:subcarrier_tail}. For PC-AFDM, the endpoint contribution is again $2$, and the internal jumps contribute $\sum_{r\in\mathcal{R}_m}|\Delta g_{m,r}|^2$. This proves the second line of \eqref{eq:subcarrier_tail}. Substituting \eqref{eq:jump_magnitude} gives the explicit sine form. Averaging over $m$ via \eqref{eq:average_esd_def} gives \eqref{eq:esd_tail}.

The same argument applies to the negative frequency interval $(-\infty,-F]$. Its $1/F$ coefficient is identical, while the oscillatory cross terms remain $O(F^{-2})$. Adding the positive and negative frequency tails gives \eqref{eq:two_sided_esd_tail}.
\end{proof}

Theorem~\ref{thm:far_out_tail} identifies the source of the asymptotic difference between the high-frequency spectral tails of PC-AFDM and SFDM. The common term comes from finite block truncation, while the additional term in PC-AFDM comes from internal jumps in the complex envelope. Three regimes follow directly:
\begin{itemize}[leftmargin=0.5cm]
\item If $\alpha\le 1/(2N)$, no internal wrapping occurs, and the high-frequency tail coefficients of PC-AFDM and SFDM coincide.
\item If $\alpha=1/(2k)$ for some $k\in\mathbb{Z}_{>0}$, internal wrapping may occur, but every internal jump vanishes. The high-frequency tail coefficients again coincide.
\item For generic values of $\alpha$, at least one internal jump is nonzero, and the PC-AFDM high-frequency tail coefficient is strictly larger than that of SFDM.
\end{itemize}

We emphasize that Theorem~\ref{thm:far_out_tail} characterizes the high-frequency tail only. It does not imply that the full OOBE integral is determined solely by the jump term, because the full OOBE also includes the spectrum near the nominal band. The next subsection therefore derives exact continuous-time spectra for evaluating the full OOBE.

\subsection{Exact Spectral Evaluation of OOBE}
\label{subsec:exact_oobe_evaluation}

The full OOBE cannot be obtained from the high-frequency tail coefficient alone. It also depends on the spectrum near the nominal band, where the asymptotic expansion is not sufficient. For this reason, we evaluate the finite block spectra \cite{oppenheim1997signals} exactly. For SFDM, the phase is affine on each sampling interval, so the spectrum is a sum of sinc terms. For PC-AFDM, the phase is quadratic on each interval with a fixed wrap count, so the spectrum is a sum of Fresnel integral terms.

\begin{prop}[Exact subcarrier spectra]
\label{prop:exact_subcarrier_spectra}
For SFDM, the $m$-th subcarrier spectrum is
\begin{equation}
\begin{aligned}
G_m^{(\mathrm{step})}(f)
&=\frac{1}{B}\sum_{n=0}^{N-1}
\exp\left[j2\pi\left(\phi_{m,n}^{(0)}-\frac{fn}{B}\right)\right]\\
&\quad\times \exp\left[j\pi\frac{f_{m,n}^{(\mathrm{step})}-f}{B}\right]
\operatorname{sinc}\left(\frac{f_{m,n}^{(\mathrm{step})}-f}{B}\right),
\end{aligned}
\label{eq:step_spectrum_sinc}
\end{equation}
where $\phi_{m,n}^{(0)}=\phi_m^{(\mathrm{step})}(n/B)$ and $\operatorname{sinc}(x)=\sin(\pi x)/(\pi x)$.

For PC-AFDM, let $0=\vartheta_{m,0}<\vartheta_{m,1}<\cdots<\vartheta_{m,J_m}<\vartheta_{m,J_m+1}=T$ denote the ordered boundary sequence formed by the internal wrapping instants of the $m$-th subcarrier, and let $q_{m,j}$ be the constant wrap count on $[\vartheta_{m,j},\vartheta_{m,j+1})$. Then
\begin{equation}
\begin{aligned}
G_m^{(\mathrm{pc})}(f)
&=\sum_{j=0}^{J_m}\int_{\vartheta_{m,j}}^{\vartheta_{m,j+1}}
\exp\Bigg\{j2\pi\Bigg[\frac{K}{2}t^2\\
&\qquad\qquad+\left(\frac{m}{T}-q_{m,j}B-f\right)t\Bigg]\Bigg\}dt.
\end{aligned}
\label{eq:pc_spectrum_integral}
\end{equation}
For $K>0$, this is equivalently
\begin{equation}
\begin{aligned}
G_m^{(\mathrm{pc})}(f)
&=\frac{1}{\sqrt{2K}}\sum_{j=0}^{J_m}
\exp\left(-j\pi\frac{\beta_{m,j}^2(f)}{K}\right)\\
&\quad\times\left[
\mathcal{F}_{\mathrm{Fr}}\left(u_{m,j+1}(f)\right)
-\mathcal{F}_{\mathrm{Fr}}\left(u_{m,j}(f)\right)\right],
\end{aligned}
\label{eq:pc_spectrum_fresnel}
\end{equation}
where
\[
\begin{aligned}
\beta_{m,j}(f) &= \frac{m}{T}-q_{m,j}B-f,\\
u_{m,j}(f) &= \sqrt{2K}\left(\vartheta_{m,j}+\frac{\beta_{m,j}(f)}{K}\right),
\end{aligned}
\]
and $\mathcal{F}_{\mathrm{Fr}}(u)=\int_0^u e^{j\pi v^2/2}dv$ is the Fresnel integral \cite{bracewell1986fourier}.
\end{prop}

\begin{proof}
For SFDM, the phase is affine on $I_n$, namely
\[
\phi_m^{(\mathrm{step})}(t)
=\phi_{m,n}^{(0)}+f_{m,n}^{(\mathrm{step})}\left(t-\frac{n}{B}\right).
\]
Substituting this expression into $G_m^{(\mathrm{step})}(f)$ gives
\[
\begin{aligned}
G_m^{(\mathrm{step})}(f)
&=\sum_{n=0}^{N-1}\int_{n/B}^{(n+1)/B}
\exp\Bigg\{j2\pi\Bigg[\phi_{m,n}^{(0)}\\
&\quad+f_{m,n}^{(\mathrm{step})}\left(t-\frac{n}{B}\right)-ft\Bigg]\Bigg\}dt.
\end{aligned}
\]
Using $t=n/B+t'$ yields
\[
\begin{aligned}
G_m^{(\mathrm{step})}(f)
&=\sum_{n=0}^{N-1}\exp\left[j2\pi\left(\phi_{m,n}^{(0)}-\frac{fn}{B}\right)\right]\\
&\quad\times\int_0^{1/B}e^{j2\pi(f_{m,n}^{(\mathrm{step})}-f)t'}dt'.
\end{aligned}
\]
The identity
\[
\int_0^{1/B}e^{j2\pi at'}dt'
=\frac{1}{B}e^{j\pi a/B}\operatorname{sinc}\left(\frac{a}{B}\right)
\]
proves \eqref{eq:step_spectrum_sinc}.

For PC-AFDM, the wrap count is constant on $[\vartheta_{m,j},\vartheta_{m,j+1})$, so $\phi_m^{(\mathrm{pc})}(t)=Kt^2/2+(m/T-q_{m,j}B)t$. Substitution gives \eqref{eq:pc_spectrum_integral}. With $\beta_{m,j}(f)=m/T-q_{m,j}B-f$, completing the square gives
\[
\frac{K}{2}t^2+\beta_{m,j}(f)t
=\frac{K}{2}\left(t+\frac{\beta_{m,j}(f)}{K}\right)^2
-\frac{\beta_{m,j}^2(f)}{2K}.
\]
Applying $u=\sqrt{2K}(t+\beta_{m,j}(f)/K)$ to each segment gives \eqref{eq:pc_spectrum_fresnel}.
\end{proof}

The exact forms in Proposition~\ref{prop:exact_subcarrier_spectra} are used to evaluate the ESD and the full OOBE without relying on a sampled time-domain approximation.
\section{Receiver Sampling Under Continuous-Time Propagation}
\label{sec:sampling_postprocessing}

Sections~\ref{sec:ct_realizations} and \ref{sec:discontinuity_spectral_leakage} show that PC-AFDM and SFDM have the same values at the sampling instants but different continuous-time trajectories. This section studies the same distinction after continuous-time propagation and sampling at the receiver \cite{Wu2011OversampledOFDM}. The main result is that the two realizations induce the same sampled channel matrix when all path delays are integer multiples of the sampling period, but generally induce different sampled channel matrices under fractional delays.

\subsection{Continuous-Time Channel and Receiver Sampling}
\label{subsec:ct_channel_sampling}

Let $s^{(\xi)}(t)$, $\xi\in\{\mathrm{pc},\mathrm{step}\}$, denote the continuous-time block over $0\le t<T$. Before transmission, a chirp periodic prefix (CPP) \cite{cao2025agile} of duration $T_{\mathrm{cpp}}$ is prepended:
\begin{equation}
s_{\mathrm{tx}}^{(\xi)}(t)=
\begin{cases}
s^{(\xi)}(t+T)e^{-j2\pi c_1N(N+2Bt)}, & -T_{\mathrm{cpp}}\le t<0,\\
s^{(\xi)}(t), & 0\le t<T.
\end{cases}
\label{eq:ct_cpp_def}
\end{equation}
The prefix duration is assumed to exceed the maximum path delay. The received signal over a continuous-time doubly selective channel with $L$ paths is
\begin{equation*}
r^{(\xi)}(t)=\sum_{\ell=1}^{L}h_{\ell}s_{\mathrm{tx}}^{(\xi)}(t-\tau_{\ell})e^{j2\pi\nu_{\ell}t}+w(t),
\label{eq:ct_ltv_channel}
\end{equation*}
where $h_{\ell}$, $\tau_{\ell}$, and $\nu_{\ell}$ are the path gain, delay, and Doppler shift of the $\ell$-th path, and $w(t)$ denotes additive noise. After CPP removal, the receiver samples $r^{(\xi)}(t)$ at $t_n=n/B$.

Define the $m$-th subcarrier waveform on the data block as
\[
u_m^{(\xi)}(t)=\frac{1}{\sqrt{N}}e^{j2\pi(c_2m^2+\phi_m^{(\xi)}(t))},
\quad 0\le t<T.
\]
The corresponding waveform after CPP insertion is
\[
u_{m,\mathrm{tx}}^{(\xi)}(t)=
\begin{cases}
u_m^{(\xi)}(t+T)e^{-j2\pi c_1N(N+2Bt)}, & -T_{\mathrm{cpp}}\le t<0,\\
u_m^{(\xi)}(t), & 0\le t<T.
\end{cases}
\]
Thus, $s_{\mathrm{tx}}^{(\xi)}(t)=\sum_{m=0}^{N-1}x[m]u_{m,\mathrm{tx}}^{(\xi)}(t)$. The sampled noiseless input-output relation is
\begin{equation*}
\bar{\bm r}^{(\xi)}=\bm H_{\xi}\bm x,
\label{eq:sampled_channel_matrix_relation}
\end{equation*}
where $\bar{\bm r}^{(\xi)}=[\bar r^{(\xi)}[0],\ldots,\bar r^{(\xi)}[N-1]]^T$ and
\begin{equation}
[\bm H_{\xi}]_{n,m}
=\sum_{\ell=1}^{L}h_{\ell}e^{j2\pi\nu_{\ell}t_n}
u_{m,\mathrm{tx}}^{(\xi)}(t_n-\tau_{\ell}).
\label{eq:effective_sampled_channel_matrix}
\end{equation}
Thus, the sampled channel matrix depends on the continuous-time realization through the delayed waveform values in \eqref{eq:effective_sampled_channel_matrix}. Define $\Delta\bm H=\bm H_{\mathrm{step}}-\bm H_{\mathrm{pc}}$. Its entries are
\begin{equation}
\begin{aligned}
{}[\Delta\bm H]_{n,m}
&=\sum_{\ell=1}^{L}h_{\ell}e^{j2\pi\nu_{\ell}t_n}\\
&\quad\times
\left[
u_{m,\mathrm{tx}}^{(\mathrm{step})}(t_n-\tau_{\ell})
-u_{m,\mathrm{tx}}^{(\mathrm{pc})}(t_n-\tau_{\ell})
\right].
\end{aligned}
\label{eq:delta_H_entries}
\end{equation}
Equation~\eqref{eq:delta_H_entries} shows that equality at the sampling instants before propagation does not automatically imply equality after continuous-time propagation.

\subsection{Sampling Equivalence and Fractional Delay}
\label{subsec:receiver_equivalence_fractional_delay}

\begin{prop}[Sampling equivalence under integer delays]
\label{prop:receiver_grid_equivalence}
Write each path delay as $\tau_{\ell}=(d_{\ell}+\epsilon_{\ell})/B$, where ${d_{\ell}\in\mathbb{Z}_{\ge 0}}$ and $\epsilon_{\ell}\in[0,1)$. Suppose every path delay is aligned with the sampling grid, i.e., $\epsilon_{\ell}=0$ for all $\ell$, and the CPP is long enough to cover all delays. Then
\begin{equation}
\Delta\bm H=\bm 0.
\label{eq:receiver_equivalence_integer_delay}
\end{equation}
If at least one path has a fractional delay, i.e., $\epsilon_{\ell}\in(0,1)$ for at least one $\ell$, then
\begin{equation}
\begin{aligned}
{}[\Delta\bm H]_{n,m}
&=\sum_{\ell=1}^{L}h_{\ell}e^{j2\pi\nu_{\ell}t_n}\\
&\hspace{-1em}\times
\left[
u_{m,\mathrm{tx}}^{(\mathrm{step})}\left(\frac{n-d_{\ell}-\epsilon_{\ell}}{B}\right)
-u_{m,\mathrm{tx}}^{(\mathrm{pc})}\left(\frac{n-d_{\ell}-\epsilon_{\ell}}{B}\right)
\right],
\end{aligned}
\label{eq:fractional_delay_mismatch_matrix}
\end{equation}
which is generally nonzero.
\end{prop}

\begin{proof}
When $\epsilon_{\ell}=0$, the argument $t_n-\tau_{\ell}=(n-d_{\ell})/B$ lies on the sampling grid of the data block or the CPP. By \eqref{eq:matrix_equivalence}, PC-AFDM and SFDM have identical samples on the data block. The CPP construction in \eqref{eq:ct_cpp_def} applies the same deterministic phase factor to both realizations, so their CPP samples are also identical at the sampling instants. Hence,
\[
u_{m,\mathrm{tx}}^{(\mathrm{step})}(t_n-\tau_{\ell})
=u_{m,\mathrm{tx}}^{(\mathrm{pc})}(t_n-\tau_{\ell})
\]
for all $m$, $n$, and $\ell$, which proves \eqref{eq:receiver_equivalence_integer_delay}.

For a fractional delay, the argument $(n-d_{\ell}-\epsilon_{\ell})/B$ is not a sampling instant. The received sample then depends on the waveform between sampling instants, and \eqref{eq:fractional_delay_mismatch_matrix} follows from \eqref{eq:delta_H_entries}. Since PC-AFDM and SFDM differ between sampling instants in general, the sampled channel matrices are generally different.
\end{proof}

Proposition~\ref{prop:receiver_grid_equivalence} shows that equality at the transmitter samples is sufficient for receiver equivalence only when all propagation delays remain on the sampling grid. With fractional delays, the receiver observes the values of the delayed waveform between sampling instants. The resulting mismatch depends on the fractional delay locations. By Proposition~\ref{prop:local_pc_sfdm_relation}, the difference is small on intervals where no PC-AFDM wrapping occurs. Near an internal wrapping instant, the difference can be dominated by the PC-AFDM jump term.

\subsection{Receiver Mismatch Metric}
\label{subsec:receiver_mismatch_metric}

After standard AFDM demodulation, the DAF-domain noiseless received vector is $\bar{\bm z}^{(\xi)}=\bm A_{\mathrm{IDAFT}}^H\bar{\bm r}^{(\xi)}=\bm G_{\xi}\bm x$, where $\bm G_{\xi}=\bm A_{\mathrm{IDAFT}}^H\bm H_{\xi}$. The DAF-domain mismatch is $\Delta\bm G=\bm G_{\mathrm{step}}-\bm G_{\mathrm{pc}}=\bm A_{\mathrm{IDAFT}}^H\Delta\bm H$. Since $\bm A_{\mathrm{IDAFT}}$ is unitary, $\|\Delta\bm G\|_{\mathrm F}^2=\|\Delta\bm H\|_{\mathrm F}^2$. For $\mathbb{E}[\bm x\bm x^H]=\bm I$, the mean square mismatch between the two noiseless DAF-domain received vectors is
\[
\mathbb{E}_{\bm x}\left[\|\bar{\bm z}^{(\mathrm{step})}-\bar{\bm z}^{(\mathrm{pc})}\|^2\right]
=\|\Delta\bm H\|_{\mathrm F}^2.
\]
Accordingly, we use the following normalized channel matrix mismatch:
\begin{equation}
\mathrm{NMSE}_{\mathrm{ch}}
=
\frac{\|\Delta\bm H\|_{\mathrm F}^2}{\|\bm H_{\mathrm{pc}}\|_{\mathrm F}^2}
=
\frac{\|\Delta\bm G\|_{\mathrm F}^2}{\|\bm G_{\mathrm{pc}}\|_{\mathrm F}^2}.
\label{eq:channel_nmse_def}
\end{equation}
This metric measures the normalized mismatch between the channel matrices induced by the two continuous-time realizations.

\subsection{Windowing Baseline}
\label{subsec:windowing_baseline}

Windowing is a common method for reducing spectral leakage. In the present problem, it should be interpreted as a postprocessing baseline rather than as a continuous-time realization of the discrete AFDM block. Let $\omega(t)$ be a continuous and piecewise continuously differentiable window over $[0,T]$, and define ${\widetilde g_m^{(\mathrm{pc})}(t)=\omega(t)g_m^{(\mathrm{pc})}(t)}$. At an internal PC-AFDM wrapping instant $t_{m,r}$, the jump of the windowed basis waveform is
\[
{\Delta \widetilde g_{m,r}=\omega(t_{m,r})\Delta g_{m,r}.}
\]
Thus, a continuous window scales the internal jump by the window value at the wrapping instant. It does not remove a nonzero internal jump unless $\omega(t_{m,r})=0$.

The same relation appears in the high-frequency tail. Define
\[
{
C_{m,\omega}=|\omega(0^{+})|^2+|\omega(T^{-})|^2
+\sum_{r\in\mathcal{R}_m}|\omega(t_{m,r})|^2|\Delta g_{m,r}|^2.
}
\]
Using the same integration by parts argument as in Lemma~\ref{lem:high_frequency_expansion}, {for the positive-frequency tail}, we have
\[
\int_F^{\infty}|\widetilde G_m^{(\mathrm{pc})}(f)|^2df
=\frac{{C_{m,\omega}}}{4\pi^2F}+O(F^{-2}),
\quad F\to\infty.
\]
The endpoint terms are weighted by ${\omega(0^{+})}$ and ${\omega(T^{-})}$, while the internal jump terms are weighted by ${\omega(t_{m,r})}$.

Windowing can also change the discrete AFDM block. Let ${\omega_n=\omega(n/B)}$ and ${\bm W_\omega=\operatorname{diag}(\omega_0,\ldots,\omega_{N-1})}$. If the receiver applies the standard AFDM demodulator to the windowed samples, the noiseless output is
\[
{\widetilde{\bm y}=\bm A_{\mathrm{IDAFT}}^H\bm W_\omega\bm A_{\mathrm{IDAFT}}\bm x.}
\]
For $\mathbb{E}[\bm x\bm x^H]=\bm I$, the resulting sample distortion induced by the window is
\begin{equation}
\mathrm{EVM}_{\mathrm{win}}^2
=\frac{1}{N}\mathbb{E}\left[\left\|\widetilde{\bm y}-\bm x\right\|^2\right]
=\frac{1}{N}\sum_{n=0}^{N-1}|{\omega_n}-1|^2.
\label{eq:window_evm}
\end{equation}
{This distortion metric is used only for the windowing tradeoff and is separate from the LMMSE receiver EVM.}
Therefore, ordinary edge windowing can reduce OOBE, but the reduction is obtained by changing endpoint behavior and, in general, by modifying the sampled AFDM block. If the window is constrained by ${\omega(n/B)=1}$ for all $n$, then the AFDM samples are preserved, but the window cannot generally eliminate the internal PC-AFDM jumps because their locations depend on $m$ and the wrapping boundary index. SFDM instead removes the internal jump term through the continuous-time waveform construction while preserving the discrete AFDM block.
\section{Numerical Results}
\label{sec:numerical_results}

This section validates the spectral and receiver consequences of the continuous-time realization. The waveform illustration has already been given in Section~\ref{subsec:sfdm_realization}. The simulations here focus on performance metrics and postprocessing baselines. We first verify the high-frequency spectral tail coefficient in Theorem~\ref{thm:far_out_tail}. We then evaluate the ESD and the full OOBE over the nominal out-of-band region. Next, we examine the receiver sensitivity caused by fractional delay mismatch. Finally, we use edge windowing as a postprocessing baseline to quantify the OOBE and distortion tradeoff.

A normalized bandwidth $B=1$ is used throughout the simulations. This normalization does not affect the relative spectral or receiver comparisons between PC-AFDM and SFDM. Unless otherwise specified, we set $N=64$ and $c_2=0$. The normalized chirp rate parameter is $\alpha=c_1N$. The average ESD and OOBE are evaluated according to Section~\ref{sec:discontinuity_spectral_leakage}. For FFT based spectral evaluation, a zero padded FFT is used on an oversampled continuous-time grid, and spectral integrals are approximated by numerical integration over the stated frequency ranges. The same numerical integration settings are used for PC-AFDM and SFDM in each comparison. For the receiver experiments, the continuous-time channel is sampled after propagation as described in Section~\ref{sec:sampling_postprocessing}. The purpose is to quantify how different trajectories between sampling instants affect spectrum containment and receiver robustness.

\begin{figure}[!t]
    \centering
    \includegraphics[width=0.98\linewidth]{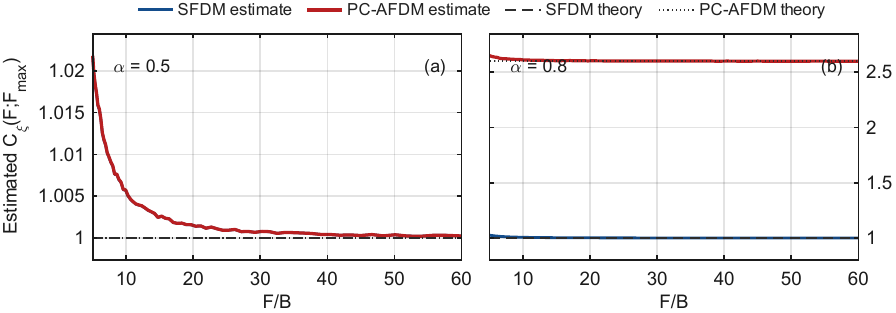}
    \caption{Verification of the high-frequency tail coefficient. The normalized finite interval coefficient $\widehat C_{\xi}(F;F_{\max})$ is plotted against $F/B$. At the special chirp rate $\alpha=0.5$, PC-AFDM and SFDM converge to the same normalized endpoint contribution because the internal jump contribution vanishes. At the generic chirp rate $\alpha=0.8$, PC-AFDM converges to a larger coefficient due to nonzero internal jumps, while SFDM remains at the common normalized endpoint contribution.}
    \label{fig:tail_convergence}
\end{figure}

\begin{figure}[!t]
    \centering
    \includegraphics[width=0.98\linewidth]{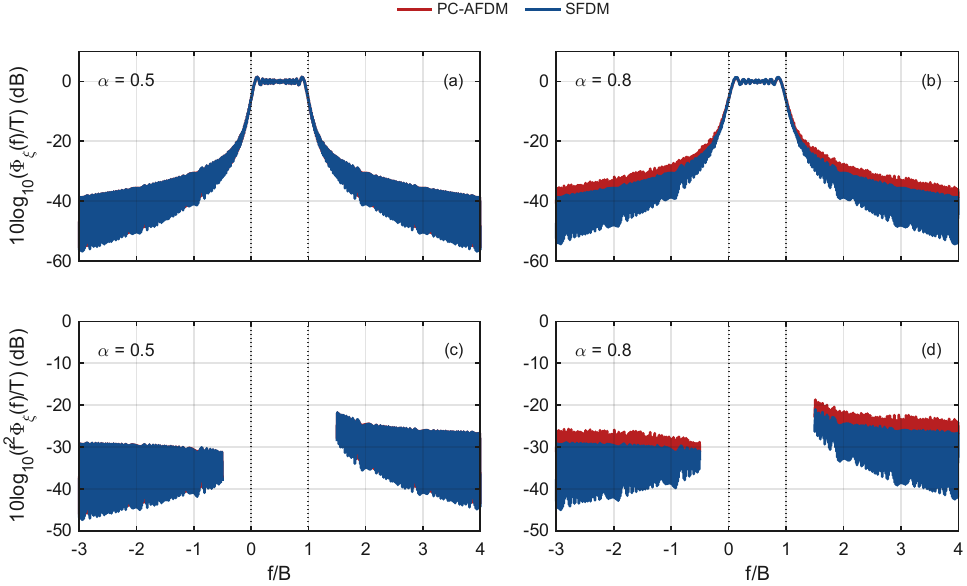}
    \caption{ESD comparison between PC-AFDM and SFDM. Top row: normalized ESD over $f/B\in[-3,4]$. Bottom row: $f^2$ compensated ESD over $f/B\in[-3,-0.5]\cup[1.5,4]$. The spectra are almost identical at the special chirp rate $\alpha=0.5$, while PC-AFDM has a larger compensated ESD at the generic chirp rate $\alpha=0.8$.}
    \label{fig:spec_density_representative}
\end{figure}

\begin{figure}[!t]
    \centering
    \includegraphics[width=0.75\linewidth]{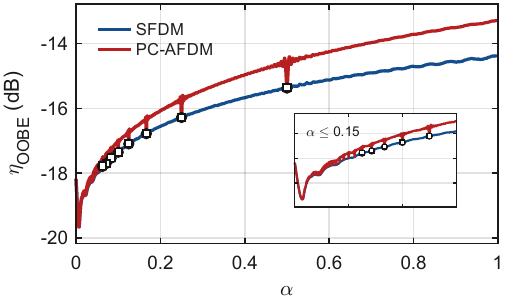}
    \caption{Average OOBE ratio $\eta_{\mathrm{OOBE}}$ versus the normalized chirp rate parameter $\alpha$. SFDM reduces the OOBE of PC-AFDM over most generic chirp rate values. Around the special continuity points $\alpha=1/(2k)$, the gap is reduced because the PC-AFDM phase resets do not produce nonzero envelope jumps.}
    \label{fig:oobe_alpha}
\end{figure}

\begin{figure*}[!t]
    \centering
    \includegraphics[width=0.7\textwidth]{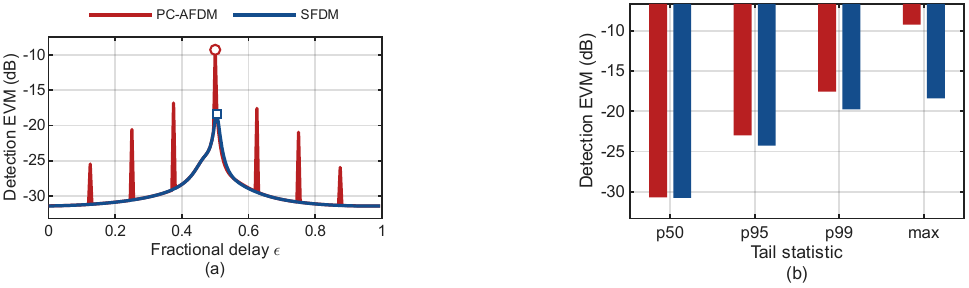}
    \caption{LMMSE receiver EVM under single-path fractional delay mismatch. The true path delay is $(d+\epsilon)/B$ with $d=4$, while the receiver uses a mismatched delay $(d+\epsilon+\Delta\epsilon)/B$ with $\Delta\epsilon=0.005$. (a) Receiver EVM versus the fractional delay $\epsilon$. (b) Tail statistics over all evaluated $\epsilon$ values. SFDM reduces the sharp EVM peaks caused by the discontinuous trajectory between sampling instants of PC-AFDM. The main gain is observed in the high percentile and worst case EVM rather than in an unconditional average improvement.}
    \label{fig:lmmse_detection_evm}
\end{figure*}

\begin{figure}[!t]
    \centering
    \includegraphics[width=0.6\linewidth]{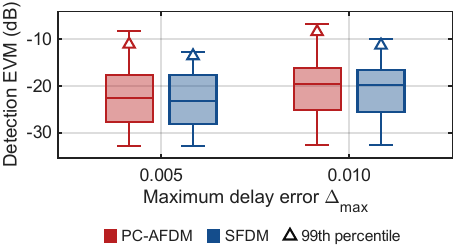}
    \caption{Receiver EVM distribution under random three-path channels with fractional delay estimation errors. The box plots summarize the EVM over random channel realizations, and the triangle markers indicate the 99th percentile. SFDM mainly reduces the high EVM tail, which is the relevant regime when fractional delay uncertainty places PC-AFDM samples near internal wrapping discontinuities.}
    \label{fig:multipath_evm_boxplot}
\end{figure}

\begin{figure*}[!t]
    \centering
    \includegraphics[width=0.78\textwidth]{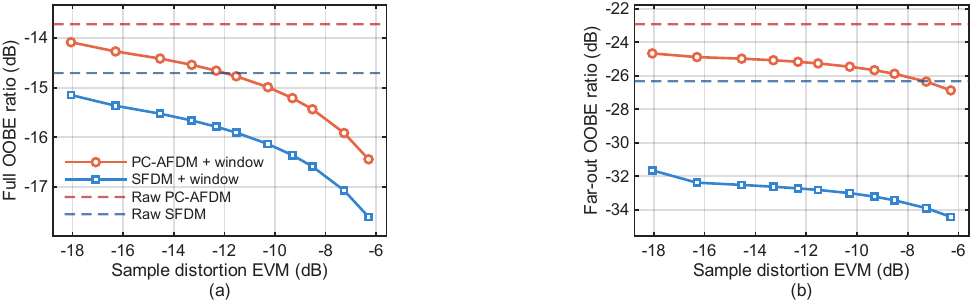}
    \caption{OOBE and sample distortion tradeoff for edge windowed PC-AFDM and SFDM at $\alpha=0.8$. (a) Full OOBE over the nominal out-of-band region. (b) Far-out OOBE. Edge windowing can reduce OOBE for both realizations, but the reduction is obtained by perturbing the AFDM samples. Raw SFDM already has a lower OOBE level than raw PC-AFDM, and windowing provides an additional postprocessing option at the cost of sample distortion rather than a replacement for the proposed continuous-time realization.}
    \label{fig:window_tradeoff}
\end{figure*}

\subsection{Spectral Tail and OOBE}
\label{subsec:oobe_performance}

We first verify the high-frequency spectral tail coefficient in Theorem~\ref{thm:far_out_tail}. Since the theorem characterizes the high-frequency tail rather than the full OOBE integral, we evaluate the normalized tail coefficient over a finite frequency interval,
\begin{equation*}
\widehat C_{\xi}(F;F_{\max})
=
\frac{\pi^2\int_{F<|f|<F_{\max}}\Phi^{(\xi)}(f)\,df}
{1/F-1/F_{\max}}.
\label{eq:finite_band_tail_coefficient}
\end{equation*}
According to Theorem~\ref{thm:far_out_tail}, $\widehat C_{\xi}(F;F_{\max})$ approaches {the normalized coefficient multiplying $1/(\pi^2F)$ in the two sided high-frequency tail} as $F$ increases. We evaluate $\alpha=0.5$ and $\alpha=0.8$, which represent a special continuity point and a generic chirp rate, respectively.

Fig.~\ref{fig:tail_convergence} confirms the tail coefficient analysis. For $\alpha=0.5$, PC-AFDM satisfies the continuity condition in Theorem~\ref{thm:pc_continuity_criterion}. Hence, its internal jump contribution vanishes and both realizations converge to the same normalized coefficient. For $\alpha=0.8$, SFDM still converges to the normalized endpoint contribution, whereas PC-AFDM converges to a larger normalized value caused by the internal envelope jumps. This verifies that the additional high-frequency tail of PC-AFDM is not merely a finite block endpoint effect.

We next examine the full spectral behavior. The high-frequency coefficient isolates the asymptotic tail, while the OOBE metric in \eqref{eq:oobe_ratio_def} integrates the complete nominal out-of-band region and therefore includes both near-band leakage and high-frequency spectral tails.

Fig.~\ref{fig:spec_density_representative} shows representative ESD plots. At $\alpha=0.5$, the spectra of PC-AFDM and SFDM are nearly identical, which is consistent with the disappearance of internal jumps. At $\alpha=0.8$, PC-AFDM has stronger high-frequency leakage, and this difference becomes clearer in the $f^2$ compensated plots. The compensated spectra therefore provide a direct visual check of the $1/f^2$ tail coefficient derived in Theorem~\ref{thm:far_out_tail}. The full OOBE behavior is evaluated separately below, because the full OOBE also includes near-band leakage.

Fig.~\ref{fig:oobe_alpha} shows the normalized OOBE ratio as a function of $\alpha$. SFDM yields lower OOBE than PC-AFDM over most of the evaluated range. In the small $\alpha$ regime, the two curves are close because internal wrapping occurs only for a limited number of high-index subcarriers. As $\alpha$ increases, internal wrapping events become more frequent and more widely distributed across the block, which creates additional spectral leakage in PC-AFDM. Around the special points $\alpha=1/(2k)$, the gap is reduced because the corresponding phase resets are integer multiples of $2\pi$. These results agree with the continuity condition in Theorem~\ref{thm:pc_continuity_criterion}.

\subsection{Single-Path Fractional Delay Mismatch}
\label{subsec:lmmse_detection_evm}

The preceding spectral results validate the OOBE advantage of SFDM. We next examine whether the different continuous-time realizations also affect receiver robustness. This question is relevant because PC-AFDM and SFDM are equivalent only at the transmit sampling instants. After a continuous-time fractional delay channel, the receiver samples delayed waveform values that generally lie between the original sampling instants.

We consider a single-path channel with gain $h=1$, zero Doppler shift, and delay $\tau=(d+\epsilon)/B$, where $d=4$ and $\epsilon\in[0,1)$. The receiver forms a mismatched parametric sampled channel matrix using $\epsilon+\Delta\epsilon$, while the true received signal is generated using $\epsilon$. Unless otherwise specified, $\Delta\epsilon=0.005$ and the SNR is $35$ dB. The sampled noise variance is normalized with respect to unit average received signal power, and the noise variance used in the LMMSE receiver is $\sigma_w^2=10^{-\mathrm{SNR}/10}$. For each realization $\xi\in\{\mathrm{pc},\mathrm{step}\}$, the noiseless sampled input-output relation is $\bar{\bm r}^{(\xi)}=\bm H_{\xi}(\epsilon)\bm x$. The LMMSE receiver is built from the mismatched matrix $\widehat{\bm H}_{\xi}=\bm H_{\xi}(\epsilon+\Delta\epsilon)$ as
\[
\bm W_{\xi}
=
\left(
\widehat{\bm H}_{\xi}^{H}\widehat{\bm H}_{\xi}
+\sigma_w^2\bm I
\right)^{-1}
\widehat{\bm H}_{\xi}^{H}.
\]
The corresponding receiver EVM is computed analytically as
\begin{equation}
\mathrm{EVM}_{\xi}^{2}
=
\frac{1}{N}
\left\|
\bm W_{\xi}\bm H_{\xi}(\epsilon)-\bm I
\right\|_{\mathrm F}^{2}
+
\frac{\sigma_w^2}{N}
\left\|
\bm W_{\xi}
\right\|_{\mathrm F}^{2}.
\label{eq:sim_detection_evm}
\end{equation}

Fig.~\ref{fig:lmmse_detection_evm} shows the resulting receiver EVM. PC-AFDM exhibits sharp EVM peaks at certain fractional delay values. These peaks occur because a small delay perturbation can move the receiver sampling point across a wrapping-induced envelope jump. SFDM varies more smoothly because its phase is accumulated continuously across sampling intervals. The distribution summary confirms that the difference is most pronounced in the high percentile and worst case EVM. Therefore, the receiver advantage of SFDM should be interpreted as reduced sensitivity to fractional delay uncertainty, not as a universal average EVM gain under all possible channel realizations.

\subsection{Random Multipath Fractional Delay Validation}
\label{subsec:random_multipath_validation}

The single-path experiment isolates the discontinuity mechanism. To test whether the effect persists in a less controlled channel, we further consider random three-path channels. For the random multipath experiment in Fig.~\ref{fig:multipath_evm_boxplot}, we use $L=3$ and fix $\alpha=0.8$. The delay of the $\ell$-th path is $\tau_\ell=(d_\ell+\epsilon_\ell)/B$, where $d_\ell$ is uniformly drawn without replacement from $\{1,2,\ldots,7\}$ and $\epsilon_\ell\sim\mathcal U(0,1)$. The normalized Doppler shift is $\widetilde{\nu}_\ell=\nu_\ell/B\sim\mathcal U[-0.03,0.03]$. Thus, the sampled Doppler phase is $e^{j2\pi\widetilde{\nu}_\ell n}=e^{j2\pi\nu_\ell t_n}$ with $t_n=n/B$. Equivalently, the block normalized Doppler $\kappa_\ell=\nu_\ell T$ satisfies $\kappa_\ell\sim\mathcal U[-0.03N,0.03N]$.

The delay estimation error $\delta_\ell$ is measured in seconds, and its normalized value satisfies $B\delta_\ell\sim\mathcal U[-\Delta_{\max},\Delta_{\max}]$, where $\Delta_{\max}\in\{0.005,0.01\}$. The receiver forms the mismatched parametric channel matrix using $\widehat{\tau}_\ell=\tau_\ell+\delta_\ell$ while keeping the path gains and Doppler shifts unchanged.

The unnormalized complex path gains are generated as independent circularly symmetric complex Gaussian variables $\widetilde h_\ell\sim\mathcal{CN}(0,1)$ and are then normalized as $h_\ell=\widetilde h_\ell/\sqrt{\sum_{i=1}^{L}|\widetilde h_i|^2}$, so that each channel realization satisfies $\sum_{\ell=1}^{L}|h_\ell|^2=1$.
Unless otherwise stated, we set $\mathrm{SNR}=35$ dB and use $\sigma_w^2=10^{-\mathrm{SNR}/10}$ in the LMMSE EVM calculation, with the same unit power normalization as in Section~\ref{subsec:lmmse_detection_evm}. The same channel realizations are used for PC-AFDM and SFDM to ensure a paired comparison. For each random channel realization, the LMMSE receiver EVM is computed using \eqref{eq:sim_detection_evm}.

Fig.~\ref{fig:multipath_evm_boxplot} shows the EVM distributions for two values of $\Delta_{\max}$. The median EVMs of the two realizations can be close, which is expected because PC-AFDM and SFDM remain equivalent at the sampling instants before propagation and differ only through delayed values between sampling instants. The more important difference appears in the upper tail. PC-AFDM produces larger high percentile EVM values in the random multipath ensemble, while SFDM gives a lower and more stable upper tail. This result supports the receiver interpretation of SFDM: it reduces the risk of severe mismatch-induced receiver degradation in fractional delay channels.

\subsection{OOBE and Sample Distortion Tradeoff}
\label{subsec:windowing_results}

A natural postprocessing baseline is to apply windowing to the transmitted block. Windowing can reduce spectral leakage, but it also changes the waveform samples observed by a standard AFDM receiver. This is different from SFDM, which changes the trajectory between sampling instants while preserving the IDAFT samples exactly. We therefore use edge windowing as a reference baseline for quantifying the OOBE and sample distortion tradeoff.

For the edge windowing baseline, a raised cosine taper is applied near the two block boundaries. For $\rho=0$, we set ${\omega_0(t)=1}$. For $\rho>0$, the window over one AFDM block is
\[
{
\omega_{\rho}(t)=
\begin{cases}
\dfrac{1}{2}\left[1-\cos\left(\dfrac{\pi t}{\rho/B}\right)\right],
& 0\le t<\rho/B,\\[1.0ex]
1,
& \rho/B\le t<T-\rho/B,\\[1.0ex]
\dfrac{1}{2}\left[1-\cos\left(\dfrac{\pi (T-t)}{\rho/B}\right)\right],
& T-\rho/B\le t<T.
\end{cases}
}
\]
Here $\rho$ is the edge length measured in Nyquist sampling intervals, and $\rho/B$ is the corresponding duration in seconds. The windowed PC-AFDM waveform is ${s_{\rho}^{(\mathrm{pc})}(t)=\omega_{\rho}(t)s^{(\mathrm{pc})}(t)}$ for $0\le t<T$. The same window is also applied to SFDM when plotting the SFDM windowing curve. Since ${\omega_{\rho}(t)}$ modifies the AFDM samples, we quantify the sample distortion induced by windowing by
\[
\mathrm{EVM}_{\mathrm{win}}^2(\rho)
=
\frac{1}{N}
\sum_{n=0}^{N-1}
\left|{\omega_{\rho}}\left(\frac{n}{B}\right)-1\right|^2,
\]
which is the same distortion measure as \eqref{eq:window_evm} for the raised cosine window. This metric measures only the sample distortion introduced by edge windowing, rather than the receiver output error. In Fig.~\ref{fig:window_tradeoff}, the horizontal axis for the windowing baseline is plotted as $10\log_{10}\mathrm{EVM}_{\mathrm{win}}^2(\rho)$ in dB, equivalently $20\log_{10}\mathrm{EVM}_{\mathrm{win}}(\rho)$. The sweep range is $\rho\in\{0,1,2,3,4,5,6,8,10,12,16,20\}$, where $\rho=0$ corresponds to the unwindowed waveform. The full OOBE is evaluated over $\mathbb{R}\setminus[0,B)$, while the far-out OOBE is evaluated over $f/B\in(-\infty,-0.5)\cup(1.5,\infty)$. The far-out metric is included because the internal jumps of PC-AFDM mainly affect the high-frequency spectral tail.

Fig.~\ref{fig:window_tradeoff} shows the windowing tradeoff. The raw SFDM reference has lower OOBE than raw PC-AFDM at the generic chirp rate. Applying an edge window reduces OOBE for both realizations, but moving along the tradeoff curve requires increasing sample distortion. This confirms that windowing is not a sample-preserving substitute for removing the internal jumps during waveform construction. It remains a useful engineering tool, but its benefit is obtained by perturbing the transmitted samples, whereas SFDM eliminates the internal jump source before any postprocessing is applied.

Overall, the numerical results support three conclusions. First, SFDM reduces the OOBE of PC-AFDM at generic chirp rates because it eliminates internal envelope jumps while preserving the IDAFT samples. Second, under fractional delay uncertainty, the discontinuous trajectory between sampling instants of PC-AFDM can produce large high percentile LMMSE receiver EVM, whereas SFDM gives smoother receiver behavior. Third, windowing can reduce OOBE but introduces a tradeoff with sample distortion, which distinguishes postprocessing from the proposed continuous-time waveform construction.
\section{Conclusion}
\label{sec:conclusion}

The fundamental lesson of this work is that the continuous-time realization of an AFDM block is a first-class design parameter, not an implementation afterthought. The conventional PC-AFDM inherits internal complex envelope jumps from its frequency-wrapping mechanism. These jumps are invisible to the Nyquist-rate samples but dominate the out-of-band spectral tail and drastically degrade receiver performance under fractional-delay propagation. Until now, this structural flaw has been the hidden source of unexplained high OOBE.

SFDM eliminates these jumps without altering the IDAFT output sequence. By holding a constant midpoint frequency inside each sampling interval and accumulating phase continuously, SFDM delivers a globally continuous complex envelope while being mathematically identical to PC-AFDM at every AFDM sampling point. 

Looking forward, the stepped-frequency philosophy extends beyond AFDM. 
Any modulation transform whose discrete symbols are converted to a physical waveform leaves the inter-sample trajectory free. This freedom can be systematically exploited to improve spectral containment, sensing performance, or hardware tolerance.

\bibliographystyle{IEEEtran}
\bibliography{refs}

@INPROCEEDINGS{Bemani2021AFDM,
  author    = {Bemani, Ali and Ksairi, Nassar and Kountouris, Marios},
  booktitle = {{IEEE} ICC Workshops},
  title     = {{AFDM}: A Full Diversity Next Generation Waveform for High Mobility Communications},
  year      = {2021},
  doi       = {10.1109/ICCWorkshops50388.2021.9473655}
}

@ARTICLE{Bemani2023AFDM,
  author  = {Bemani, Ali and Ksairi, Nassar and Kountouris, Marios},
  journal = {IEEE Trans. Wireless Commun.},
  title   = {Affine Frequency Division Multiplexing for Next Generation Wireless Communications},
  year    = {2023},
  volume  = {22},
  number  = {11},
  pages   = {8214--8229},
  doi     = {10.1109/TWC.2023.3260906}
}

@ARTICLE{Bemani2024ISAC,
  author  = {Bemani, Ali and Ksairi, Nassar and Kountouris, Marios},
  journal = {IEEE Wireless Commun. Lett.},
  title   = {Integrated Sensing and Communications With Affine Frequency Division Multiplexing},
  year    = {2024},
  volume  = {13},
  number  = {5},
  pages   = {1255--1259},
  doi     = {10.1109/LWC.2024.3367178}
}

@ARTICLE{Ranasinghe2025AFDMISAC,
  author={Ranasinghe, Kuranage Roche Rayan and Seok Rou, Hyeon and Thadeu Freitas de Abreu, Giuseppe and Takahashi, Takumi and Ito, Kenta},
  journal = {IEEE Trans. Wireless Commun.},
  title   = {Joint Channel, Data, and Radar Parameter Estimation for {AFDM} Systems in Doubly-Dispersive Channels},
  year    = {2025},
  volume  = {24},
  number  = {2},
  pages   = {1602--1619},
  doi     = {10.1109/TWC.2024.3510935}
}

@ARTICLE{Wu2011OversampledOFDM,
  author  = {Wu, Jingxian and Zheng, Yahong Rosa},
  journal = {IEEE Trans. Commun.},
  title   = {Oversampled Orthogonal Frequency Division Multiplexing in Doubly Selective Fading Channels},
  year    = {2011},
  volume  = {59},
  number  = {3},
  pages   = {815--822},
  doi     = {10.1109/TCOMM.2011.121410.090655}
}

@ARTICLE{10769778,
  author  = {Rou, Hyeon Seok and de Abreu, Giuseppe Thadeu Freitas and Choi, Junil and Gonz{\'a}lez G., David and Kountouris, Marios and Guan, Yong Liang and Gonsa, Osvaldo},
  journal = {IEEE Signal Process. Mag.},
  title   = {From Orthogonal Time-Frequency Space to Affine Frequency-Division Multiplexing: A Comparative Study of Next-Generation Waveforms for Integrated Sensing and Communications in Doubly Dispersive Channels},
  year    = {2024},
  volume  = {41},
  number  = {5},
  pages   = {71--86},
  doi     = {10.1109/MSP.2024.3422653}
}

@ARTICLE{cao2025agile,
  author  = {Cao, Yewen and Shao, Yulin},
  title   = {Agile Affine Frequency Division Multiplexing},
  journal = {arXiv preprint arXiv:2512.14424},
  year    = {2025}
}

@ARTICLE{11173628,
  author  = {Yin, Haoran and Tang, Yanqun and Ni, Yuanhan and Wang, Zulin and Chen, Gaojie and Xiong, Jun and Yang, Kai and Kountouris, Marios and Guan, Yong Liang and Zeng, Yong},
  journal = {IEEE J. Sel. Areas Commun.},
  title   = {Ambiguity Function Analysis of {AFDM} Signals for Integrated Sensing and Communications},
  year    = {2026},
  volume  = {44},
  pages   = {196--211},
  doi     = {10.1109/JSAC.2025.3611936}
}

@ARTICLE{Brandes2006CancellationCarriers,
  author={Brandes, S. and Cosovic, I. and Schnell, M.},
  journal = {IEEE Commun. Lett.},
  title   = {Reduction of Out-of-Band Radiation in {OFDM} Systems by Insertion of Cancellation Carriers},
  year    = {2006},
  volume  = {10},
  number  = {6},
  pages   = {420--422},
  doi     = {10.1109/LCOMM.2006.1638602}
}

@ARTICLE{vanDeBeek2009NContinuousOFDM,
  author  = {van de Beek, Jaap and Berggren, Fredrik},
  journal = {IEEE Commun. Lett.},
  title   = {{N}-Continuous {OFDM}},
  year    = {2009},
  volume  = {13},
  number  = {1},
  pages   = {1--3},
  doi     = {10.1109/LCOMM.2009.081446}
}

@INPROCEEDINGS{Hadani2017OTFS,
  author={Hadani, R. and Rakib, S. and Tsatsanis, M. and Monk, A. and Goldsmith, A. J. and Molisch, A. F. and Calderbank, R.},
  booktitle = {{IEEE} WCNC},
  title     = {Orthogonal Time Frequency Space Modulation},
  year      = {2017},
  doi       = {10.1109/WCNC.2017.7925924}
}

@ARTICLE{Raviteja2018OTFS,
  author={Raviteja, P. and Phan, Khoa T. and Hong, Yi and Viterbo, Emanuele},
  journal = {IEEE Trans. Wireless Commun.},
  title   = {Interference Cancellation and Iterative Detection for Orthogonal Time Frequency Space Modulation},
  year    = {2018},
  volume  = {17},
  number  = {10},
  pages   = {6501--6515},
  doi     = {10.1109/TWC.2018.2860011}
}

@ARTICLE{Ouyang2016OCDM,
  author  = {Ouyang, Xing and Zhao, Jian},
  journal = {IEEE Trans. Commun.},
  title   = {Orthogonal Chirp Division Multiplexing},
  year    = {2016},
  volume  = {64},
  number  = {9},
  pages   = {3946--3957},
  doi     = {10.1109/TCOMM.2016.2594792}
}

@ARTICLE{Ebihara2016OSDM,
  author  = {Ebihara, Tadashi and Leus, Geert},
  journal = {IEEE J. Ocean. Eng.},
  title   = {Doppler-Resilient Orthogonal Signal-Division Multiplexing for Underwater Acoustic Communication},
  year    = {2016},
  volume  = {41},
  number  = {2},
  pages   = {408--427},
  doi     = {10.1109/JOE.2015.2454411}
}

@ARTICLE{Mahmoud2008AST,
  author  = {Mahmoud, Hisham A. and Arslan, H{\"u}seyin},
  journal = {IEEE Commun. Lett.},
  title   = {Sidelobes Suppression in {OFDM}-Based Spectrum Sharing Systems Using Adaptive Symbol Transition},
  year    = {2008},
  volume  = {12},
  number  = {2},
  pages   = {133--135},
  doi     = {10.1109/LCOMM.2008.071729}
}

@ARTICLE{You2014OOBE,
  author  = {You, Zihao and Fang, Juan and Lu, I-Tai},
  journal = {EURASIP J. Adv. Signal Process.},
  title   = {Out-of-Band Emission Suppression Techniques Based on a Generalized {OFDM} Framework},
  year    = {2014},
  volume  = {2014},
  number  = {1},
  pages   = {74},
  doi     = {10.1186/1687-6180-2014-74}
}

@ARTICLE{11489295,
  author  = {Li, Xiangjun and Liu, Zilong and Zhou, Zhengchun and Fan, Pingzhi},
  journal = {IEEE Trans. Wireless Commun.},
  title   = {Matched Filtering-Based Channel Estimation for {AFDM} Systems in Doubly Selective Channels},
  year    = {2026},
  volume  = {25},
  pages   = {15597--15610},
  doi     = {10.1109/TWC.2026.3683823}
}

@INCOLLECTION{Matz2011TimeVaryingChannels,
  author    = {Matz, Gerald and Hlawatsch, Franz},
  title     = {Fundamentals of Time-Varying Communication Channels},
  booktitle = {Wireless Communications over Rapidly Time-Varying Channels},
  editor    = {Hlawatsch, Franz and Matz, Gerald},
  publisher = {Academic Press},
  address   = {Oxford, U.K.},
  year      = {2011},
  pages     = {1--63},
  doi       = {10.1016/B978-0-12-374483-8.00001-7}
}

@ARTICLE{Laakso1996FractionalDelay,
  author={Laakso, T.I. and Valimaki, V. and Karjalainen, M. and Laine, U.K.},
  journal = {IEEE Signal Process. Mag.},
  title={Splitting the unit delay: {FIR}/all pass filters design}, 
  year    = {1996},
  volume  = {13},
  number  = {1},
  pages   = {30--60},
  doi     = {10.1109/79.482137}
}

@ARTICLE{221081,
  author={Gardner, F.M.},
  journal={IEEE Trans. Commun.},
  title={Interpolation in digital modems. I. Fundamentals}, 
  year={1993},
  volume={41},
  number={3},
  pages={501-507},
  doi={10.1109/26.221081}}

@article{liu2004orthogonal,
  title={Orthogonal time-frequency signaling over doubly dispersive channels},
  author={Liu, Ke and Kadous, Tamer and Sayeed, Akbar M},
  journal={IEEE Trans. Inf. Theory},
  volume={50},
  number={11},
  pages={2583--2603},
  year={2004},
  publisher={IEEE}
}

@article{shao2021federated,
  title={Federated edge learning with misaligned over-the-air computation},
  author={Shao, Yulin and G{\"u}nd{\"u}z, Deniz and Liew, Soung Chang},
  journal={IEEE Trans. Wireless Commun.},
  volume={21},
  number={6},
  pages={3951--3964},
  year={2021},
  publisher={IEEE}
}

@article{yang2026delay,
  title={Delay-Doppler Domain Channel Estimation: What if Sparsity is Unknown?},
  author={Yang, Zijian and Shao, Yulin and Hou, Fen and Ma, Shaodan},
  journal={arXiv:2605.00049},
  year={2026}
}

@article{diez2019generalized,
  title={A generalized spectral shaping method for {OFDM} signals},
  author={Diez, Luis and Cort{\'e}s, Jos{\'e} A and Ca{\~n}ete, Francisco Javier and Martos-Naya, Eduardo and Iranzo, Salvador},
  journal={IEEE Trans. Commun.},
  volume={67},
  number={5},
  pages={3540--3551},
  year={2019},
  publisher={IEEE}
}

@article{shao2024theory,
  title={A theory of semantic communication},
  author={Shao, Yulin and Cao, Qi and G{\"u}nd{\"u}z, Deniz},
  journal={IEEE Trans. Mobile Comp.},
  volume={23},
  number={12},
  pages={12211--12228},
  year={2024},
  publisher={IEEE}
}

@book{bracewell1986fourier,
  title={The {F}ourier transform and its applications},
  author={Bracewell, Ronald Newbold and Bracewell, Ronald N},
  volume={31999},
  year={1986},
  publisher={McGraw-hill New York}
}

@book{oppenheim1997signals,
  title={Signals \& systems},
  author={Oppenheim, Alan V and Willsky, Alan S and Nawab, Syed Hamid},
  year={1997},
  publisher={Pearson Educaci{\'o}n}
}

\end{document}